\documentclass[12pt]{article}
\usepackage{amsmath,amsthm,amssymb,amsfonts}
\usepackage{enumerate}
\usepackage{xcolor}
\usepackage{authblk}
\usepackage{appendix}
\usepackage[left=1in,right=1in,top=1in,bottom=1in]{geometry}

\theoremstyle{plain}
\newtheorem{theorem}{\hskip\parindent Theorem}
\newtheorem{lemma}{\hskip\parindent Lemma}
\newtheorem{corollary}{\hskip\parindent Corollary}

\theoremstyle{definition}
\newtheorem{remark}{\hskip\parindent Remark}

\numberwithin{equation}{section}

\newcommand\Ai{\operatorname{Ai}}
\newcommand\Bi{\operatorname{Bi}}
\renewcommand\Re{\operatorname{Re}}
\renewcommand\Im{\operatorname{Im}}
\newcommand\B {\mathrm{B}}

\begin{document}

\title{Full Asymptotic Expansion of Monodromy Data for the First Painlev\'{e} Transcendent: Applications to Connection Problems}

\author[$\dag$]{Wen-Gao Long\thanks{Corresponding author: longwg@hnust.edu.cn}}

\author[$\dag$]{Yun-Jiang Jiang}

\author[$\ddag$]{Yu-Tian Li}

\affil[$\dag$]{\small School of Mathematics and Statistics,
Hunan University of Science and Technology, Xiangtan, Hunan, 411201, PR China}

\affil[$\ddag$]{\small
School of Science and Engineering, The Chinese University of Hong Kong, Shenzhen, Guangdong, 518172, PR China.}

\date{}

\maketitle


\begin{abstract}
We study the full asymptotic expansion of the monodromy data ({\it i.e.}, Stokes multipliers) for the first Painlev\'{e} transcendent (PI) with large initial data or large pole parameters. Our primary approach involves refining the complex WKB method, also known as the method of uniform asymptotics, to approximate the second-order ODEs derived from PI's Lax pair with higher-order accuracy. As an application, we provide a rigorous proof of the full asymptotic expansion of the nonlinear eigenvalues proposed numerically by Bender, Komijani, and Wang. Additionally, we present the full asymptotic expansion for the pole parameters $(p_{n}, H_{n})$ corresponding to the $n$-th pole of the real tritronqu\'{e}e solution of the PI equation as $n \to +\infty$.

\end{abstract}

\vspace{5mm}

\vspace{5mm}

\section{Introduction}

Painlev\'{e} equations are a group of six nonlinear second-order ordinary differential equations
that possess the Painlev\'{e} property. This property requires that the movable singularities of the solutions must be poles, not branch points or essential singularities.
In general, the solutions of these equations do not have explicit expressions. Consequently,
asymptotic analysis is commonly used to study the behavior of solutions to Painlev\'{e} equations. This leads to a natural question: if the asymptotic approximations of a Painlev\'e function
are given in different regimes, such as the variable $t \to -\infty$ and $t \to \infty$,
what is the relation between these asymptotic formulas?
This is known as the connection formula problem in Painlev\'e equation theory,
which has received significant attention in literature;
see for example \cite{APC, DaiHu-2017, FAS-2006, Joshi-Kruskal-1992, Kapaev-Kitaev-1993, Wong-Zhang-2009-PIII, Wong-Zhang-2009-PIV}.

In this work, we consider the connection problem for the first Painlev\'{e} equation:
\begin{equation}\label{PI equation}
\frac{d^{2}y}{dt^{2}} = 6y^2 + t.
\end{equation}
This problem has been previously studied by Joshi and Kruskal~\cite{Joshi-Kruskal-1992}, as well as by Kapaev and Kitaev~\cite{Kapaev-Kitaev-1993}, yielding significant progress.
However, the problem is not yet fully solved and has been highlighted by Clarkson as an open problem on several occasions~\cite{CPA2003,CPA2006}.
Since then, further investigations have been conducted in this direction~\cite{Bender-Komijani-2015,LongLi,LongLiWang}.

The asymptotic behaviors of the PI solutions can be derived in several interesting regimes:
\begin{equation}\label{three-cases}
(\mathrm{i})\  t \to -\infty, \qquad (\mathrm{ii})\  t \to 0,  \qquad (\mathrm{iii})\  t \to p,
\end{equation}
where $p$ is a pole of the PI solution.

Situation (i) has been extensively studied by Kapaev~\cite{AAKapaev-1988} using isomonodromy theory.
According to the asymptotic behavior of $y(t)$ as $t \to -\infty$, the solutions of the PI equation can be classified into three types: Type A (oscillatory solutions), Type B (separatrix solutions), and Type C (singular solutions). These classifications are determined by the Stokes multipliers as follows:
\begin{equation}\label{eq-classifed-by-stokes}
\begin{cases}
\text{Type A (oscillatory solutions)}: & \Im s_{0}=1+s_{2}s_{3}>0;\\
\text{Type B (separatrix solutions)}: & \Im s_{0}=1+s_{2}s_{3}=0;\\
\text{Type C (singular solutions)}:  & \Im s_{0}=1+s_{2}s_{3}<0,
\end{cases}
\end{equation}
where $s_{k}$'s are the Stokes multipliers, which will be explicitly given later in \eqref{eq-Stokes-matrices}.
The detailed asymptotic behaviors of these three types of solutions are provided in Appendix~\ref{sec:AppB} for better clarity and presentation.

Situation (ii)  corresponds to the initial value problem, where the solution admits a MacLaurin series expansion:
\begin{equation}
y(t) = a + bt + O(t^2), \qquad t \to 0,
\end{equation}
with the initial data $(y(0), y'(0)) = (a, b)$ uniquely determining the solution. 

For situation (iii), as $t \to p$, where $p$ is a pole of $y(t)$, the solution admits
a Laurent series expansion:
\begin{equation}\label{eq-Laurent-series}
y(t) = \frac{1}{(t-p)^2} - \frac{p}{10}(t-p)^2 - \frac{1}{6}(t-p)^3 + H(t-p)^4 + \mathcal{O}((t-p)^5), \qquad t \to p.
\end{equation}
Here, $H$ is a free parameter, and the pair $(p, H)$ uniquely determines a PI solution.

The monodromy theory plays
a central role in the study of the connection problem of Painlev\'e transcendents.
For presentation considerations again,
we describe the monodromy theory for PI in Appendix~\ref{sec:AppA}.
In principle, there is a one-to-one correspondence between PI solutions
and Stokes multipliers: a pair of Stokes multipliers uniquely determines a PI solution and
vice versa; see Appendix \ref{sec:AppA}.
One commonly used approach in studying connection problems is
to first establish the relationship between the Stokes multipliers and specific asymptotic
approximations,  as demonstrated by Kapaev~\cite{AAKapaev-1988} for situation (i).
Then, by using the Stokes multipliers as a bridge, the behaviors of a particular PI solution
in different regimes can be connected.

For PI, we are particularly interested in two connection problems:
(I)connection problem between the asymptotic behavior at negative infinity and the initial data;
and (II)connection problem between the asymptotic behaviors at negative infinity and the pole.
To this end, one need to calculate the Stokes multipliers for situations (ii) and (iii). In situation (ii), the Stokes multipliers are studied in \cite{LongLi}. Although it is very challenging to derive the exact expression of the Stokes multipliers in terms of $a$ and $b$, the leading asymptotic approximations of the Stokes multipliers are calculated when $a$ or $b$ is large. The corresponding results are applied to build some limiting-form connection formulas between the initial data and the parameters in the asymptotic expansions at negative infinity. In situation (iii), i.e., when $t \to p$, the PI solution is expressed as a Laurent series \eqref{eq-Laurent-series}. In \cite{LongLiWang}, the authors study the leading asymptotic behavior of the Stokes multipliers $s_{k}$'s as $p$ or $H$ tends to infinity and provide an asymptotic classification of the PI solutions with respect to $(p,H)$.

In the present work, we aim to achieve a deep understanding and refined results of the above two connection problems for PI. The major step is to derive the full asymptotic expansion of the Stokes multipliers in situations (ii) large initial data and (iii) large pole parameters.

\paragraph{Large initial data case}
\

Our motivation for discussing this case comes from \cite{Bender-Komijani-2015, Bender-Komijani-Wang}, where the authors find that there exists a sequence $\{b_{n}\}$ such that the PI solutions with $y(0)=0$ and $y'(0)=b_{n}$ are the separatrix solutions. Similarly, there exists a sequence $\{a_{n}\}$ such that the PI solutions with $y(0)=a_{n}$ and $y'(0)=0$ are also the separatrix solutions.
The values $a_{n}$ and $b_{n}$ are termed as the \emph{nonlinear eigenvalues} for the PI equation.
It is stated in \cite{Bender-Komijani-Wang}, based on numerical simulations, that
\begin{equation}\label{eq-an-bn-approx-numerical}
\begin{split}
a_{n}&\sim -\left[\frac{\sqrt{3\pi}\Gamma\left(\frac{11}{6}\right)}{\Gamma\left(\frac{1}{3}\right)}\right]^{\frac{2}{5}}\left(n-\frac{1}{2}\right)^{\frac{2}{5}}\left[1-\frac{0.0096518}{(n-\frac{1}{2})^2}+\frac{0.0240}{(n-\frac{1}{2})^4}\right], \\
b_{n}&\sim 2\left[\frac{\sqrt{3\pi}\Gamma\left(\frac{11}{6}\right)}{\Gamma\left(\frac{1}{3}\right)}\right]^{\frac{3}{5}}\left(n-\frac{1}{6}\right)^{\frac{3}{5}}\left[1-\frac{0.00551328}{(n-\frac{1}{6})^2}+\frac{0.29334}{(n-\frac{1}{6})^3}\right]
\end{split}
\end{equation}
as $n \to \infty$.
A natural problem is to obtain the full asymptotic expansions of $a_{n}$ and $b_{n}$ as $n \to \infty$. Moreover, eq.~\eqref{eq-an-bn-approx-numerical} suggests that all the odd terms vanish and that $a_{n}$ has a full asymptotic expansion of the form:
\begin{equation}\label{an-full-expansion}
a_{n}\sim -\left[\frac{\sqrt{3\pi}\Gamma\left(\frac{11}{6}\right)}{\Gamma\left(\frac{1}{3}\right)}\right]^{\frac{2}{5}}\left(n-\frac{1}{2}\right)^{\frac{2}{5}}\sum\limits_{k=0}^{\infty}\frac{A_{k}}{\left(n-\frac{1}{2}\right)^{2k}}, \quad n \to \infty.
\end{equation}
It turns out that the answer is affirmative. To provide a rigorous proof of \eqref{an-full-expansion}, we need to calculate the full asymptotic expansion of the Stokes multipliers corresponding to the PI equation as $a \to -\infty$.

\paragraph{Large pole parameter case}
\

A similar question can be proposed in the connection problem of the PI equation between negative infinity and the pole. In \cite{LongLiWang}, the authors obtain the leading asymptotic approximation of the pole parameter $(p_{n}, H_{n})$ for the $n$-th pole of the tritronqu\'{e}e solution of the PI equation. Moreover, numerical simulation shows that (see \cite[p. 6702]{LongLiWang})
\begin{equation}
\label{pn-Hn-full-expansion}
\begin{aligned}
p_{n} &\sim 2C_{0}\left(\frac{4n-2}{\kappa^{2}(C_{0})}\right)^{\frac{4}{5}}\left[1+\frac{0.0045148}{\left(n-\frac{1}{2}\right)^{2}}\right], \\
H_{n} &\sim -\frac{1}{7}\left(\frac{4n-2}{\kappa^{2}(C_{0})}\right)^{\frac{6}{5}}\left[1-\frac{0.0081}{\left(n-\frac{1}{2}\right)^{2}}\right],
\end{aligned}
\end{equation}
as $n \to \infty$, where $C_{0}$ and $\kappa^{2}(C_{0})$ are defined in \cite[Lemma 2.1]{LongLiWang}. A natural question is whether the asymptotic behavior of $p_{n}$ and $H_{n}$ both possess full asymptotic expansions with the odd terms vanishing. To answer this question, the full asymptotic expansion of the Stokes multipliers as $p$ or $H$ tends to infinity is needed.

\

In the study of connection problems for Painlev\'e equations, a frequently employed approach is
the complex WKB method \cite{Voros-1983, Kawai-Takei-2005-book} (also referred to as the method of {\it uniform asymptotics} in \cite{APC}). This method has been widely applied  in the literature \cite{Wong-Zhang-2009-PIII, Wong-Zhang-2009-PIV, LongLiWang, LongZeng, Zeng-Zhao-2015}.
The approach typically consists of two main steps. The first step involves approximating a second-order linear ODE (derived from the Lax pair) uniformly on two adjacent Stokes lines. These uniform asymptotics often utilize special functions such as Airy or parabolic cylinder functions. The second step derives the Stokes multipliers of specific Painlev\'e functions
by using the Stokes phenomena of the special functions employed in the first step.

For the first Painlev\'e equation, when considering the cases as \( t \to 0 \) and \( t \to p \), although it is very difficult to calculate the explicit representation of the Stokes multipliers \( \{s_{k}\} \), it is possible to obtain the asymptotic behavior of \( \{s_{k}\} \) when the initial data or pole parameter is large. However, since \cite[Theorems 1 and 2]{APC} only give the leading approximation of the solutions of the Lax pair, we can only obtain the leading asymptotic behavior of the Stokes multipliers.
To get a full asymptotic expansion of the Stokes multipliers, we need to improve the complex WKB method to obtain a higher-order approximation of the second-order ODEs derived from the Lax pair. This is another motivation for this paper.

The rest of this paper is organized as follows. In Sec.~\ref{sec:main-results}, we state the main results, which are divided into three parts: (1) Higher-order approximations of the Lax pair; (2) Full asymptotic expansions of the Stokes multipliers; (3) Full asymptotic expansions of the nonlinear eigenvalues and pole parameters. The proof of the third part is also included in this section, based on the results in the first two parts. In Sec. \ref{sec:proof-lemma} and Sec. \ref{sec:proof-theorem}, we provide the proofs of the first and second parts, respectively.

\section{Main results}
\label{sec:main-results}
The main results of this paper consist of three parts. First, we obtain the higher-order approximation of the second-order ODEs derived from the Lax pair for the PI equation in case (ii) [\textit{resp.} case (iii)] stated in \eqref{three-cases}. Second, we obtain the full asymptotic expansion of the Stokes multipliers in case (ii) when the initial data [\textit{resp.} the pole parameters] tend to infinity. The third part is an application of the full asymptotic expansions of the Stokes multipliers. We derive the full asymptotic expansions of $a_{n}$, $p_{n}$, and $H_{n}$ as $n\to\infty$, namely eqs.~\eqref{an-full-expansion}
and~\eqref{pn-Hn-full-expansion}.

\subsection{Higher-order approximations of the Lax pair}

If we regard \eqref{eq-system-hat-Phi} as a linear system of ODEs for $\hat{\Phi} = (\phi_1, \phi_2)^T$, where $T$ denotes the transpose, and eliminate the equation for $\phi_2$, then system \eqref{eq-system-hat-Phi} is equivalent to a second-order ODE for $\phi_1$. We then apply the method of uniform asymptotics developed by Bassom \textit{et al.}~\cite{APC} to obtain asymptotic approximations for $\phi_{1}(\lambda)$ with a large parameter. To obtain full asymptotic expansions, rather than just the leading term, we incorporate the method of uniform asymptotics with an idea developed by Dunster~\cite{Dunster-2014}.

\paragraph{Large initial data case}
\

Assume $t\neq p$, where $p$ is any pole of the PI solution $y(t)$.
It follows from \eqref{eq-system-hat-Phi} that $Y$ satisfies a second-order ODE:
\begin{equation}\label{Schrodinger-equation-t-general}
\frac{d^{2}\phi_{1}}{d\lambda^{2}}=\left[y_{t}^{2}+4\lambda^{3}+2\lambda t-2y t-4y^{3}-\frac{y_{t}}{\lambda-y}+\frac{3}{4}\frac{1}{(\lambda-y)^2}\right]\phi_{1}.
\end{equation}
When $t=0$, we introduce a real large parameter $\xi$ defined by
\[
\xi^{\frac{6}{5}}=\frac{y'(0)^2}{4}-y(0)^3=\frac{b^2}{4}-a^3
\]
and set
\[
a=A(\xi)\xi^{\frac{2}{5}} \quad\text{and}\quad
b=B(\xi)\xi^{\frac{3}{5}},\]
where $A(\xi), B(\xi)$ are both bounded. Then $\frac{B(\xi)^2}{4}-A(\xi)^3=1$. Under the scaling transformation $\lambda=\xi^{\frac{2}{5}}z, Y(z)=\phi_{1}(\xi^{\frac{2}{5}}z)$,
eq.~\eqref{Schrodinger-equation-t-general} becomes
\begin{equation}\label{second-order-equation-w}
\frac{d^{2}Y}{dz^2}=\xi^2 \left[f(z)+\frac{g(z)}{\xi}+\frac{h(z)}{\xi^2}\right]Y,
\end{equation}
where
\[
f(z)=4(z-z_{0})(z-z_{1})(z-z_{2}),\quad
g(z)=-\frac{B(\xi)}{z-A(\xi)},\quad
h(z)=\frac{3}{4(z-A(\xi))^2}\]
with $z_{0}=-1$, $z_{1}=e^{\frac{\pi i}{3}}$ and $z_{2}=e^{-\frac{\pi i}{3}}$. It is clear that
$f(z)$, $g(z)$, and $h(z)$ are all analytic in a neighbourhood of $z=z_{1}$,
and that $f(z)$ has a simple zero at $z=z_{1}$.
Define a transformation $\zeta:=\zeta(z)$ by
\begin{equation}
\label{def-zeta}
\frac{2}{3}\zeta^{\frac{3}{2}}=\int_{z_{1}}^{z}f(t)^{\frac{1}{2}}dt,
\end{equation}
where the branches are chosen such that $\arg(z-z_{j})\in(-\pi,\pi),j=0,1,2$.
Then $\zeta(z)$ is a conformal mapping in the neighbourhood of $z=z_{1}$ and the two adjacent Stokes lines $\{\arg{\zeta}=\frac{\pi}{3},\pi\}$ emanating from $z=z_{1}$ to infinity. Moreover, a direct calculation of the integral in \eqref{def-zeta} yields
\begin{equation}\label{eq-asymp-relation-zeta-z}
\frac{2}{3}\zeta^{\frac{3}{2}}=\frac{4}{5}z^{\frac{5}{2}}+E_{0}+\mathcal{O}(z^{-\frac{1}{2}})
\end{equation}
as $z\to\infty$ uniformly for $\arg{z}\in\left[0,\frac{4\pi}{5}\right]$, where $E_{0}=\frac{3}{5}\B\left(\frac{1}{2},\frac{1}{3}\right)-\frac{\sqrt{3}i}{5}\B\left(\frac{1}{2},\frac{1}{3}\right)$, where $B(\cdot,\cdot)$ is the beta function.
Set $W(\zeta,\xi)=\left(\frac{f(z)}{\zeta}\right)^{\frac{1}{4}}Y$, then \eqref{second-order-equation-w} is transformed into
\begin{equation}\label{second-order-equation-W}
\frac{d^{2}W}{d\zeta^2}=\xi^2 \left[\zeta+\frac{\varphi(\zeta)}{\xi}+\frac{\psi(\zeta)}{\xi^2}\right]W,
\end{equation}
where
\begin{equation}\label{eq-def-varphi-psi}
\begin{split}
\varphi(\zeta)&=\frac{\zeta g(z)}{f(z)},\\
\psi(\zeta)&=\frac{5}{16\zeta^2}+\frac{\zeta[4f(z)f''(z)-5f'(z)^2]}{16f(z)^3}+\frac{\zeta h(z)}{f(z)}
\end{split}
\end{equation}
are both analytic at $\zeta=0$.
Moreover, inspired by the work of Dunster~\cite{Dunster-2014}, we have the following Lemma.
\begin{lemma}\label{lem-higher-order-approximation-ODE-1}
Let $\delta$ and $\epsilon$ be fixed small positive constants and denote
\[
\mathbb{S}:=\{\zeta\in\mathbb{C}~:~\arg{\zeta}=\frac{k\pi}{3}, k=1,3\}\]
and
\[
\mathbb{D}:=\left\{\zeta\in\mathbb{C}~:~|\zeta|\leq \delta \text{ or } \frac{\pi}{3}-\epsilon<\arg{\zeta}<\pi+\epsilon\right\}.
\]
For any $n\in\mathbb{N}$, define $\hat{\zeta}=\zeta+\mathcal{A}_{n}(\zeta,\xi)$ with
\begin{equation}\label{eq-def-An}
\mathcal{A}_{n}(\zeta,\xi)=\sum\limits_{s=1}^{2n}\frac{a_{s}(\zeta)}{\xi^{s}}
\end{equation}
satisfying
\begin{equation}\label{eq-An-coeff-compare}
\xi^2\left\{\mathcal{A}_{n}+(\zeta+\mathcal{A}_{n})(2+\mathcal{A}'_{n})\mathcal{A}'_{n}\right\}+\frac{3{\mathcal{A}''_{n}}^{2}-2(1+\mathcal{A}'_{n})\mathcal{A}'''_{n}}{4(1+\mathcal{A}'_{n})^2}=\xi\varphi(\zeta)+\psi(\zeta)+\mathcal{O}(\xi^{-2n+1})
\end{equation}
as $\xi\to+\infty$ uniformly for all $\zeta\in\mathbb{D}$.
Assume also that each $a_{s}(\zeta)$ is analytic in $\mathbb{D}$.
Then for all $s\in\mathbb{N}$, the limits
\begin{equation}\label{eq-def-alpha-s}
\alpha_{s}(A(\xi),B(\xi)):=\lim\limits_{\zeta\to\infty}\zeta^{\frac{1}{2}}a_{s}(\zeta)
\end{equation}
exist and depend only on $A(\xi)$ and $B(\xi)$.
Moreover, for any solution $W(\zeta, \xi)$ of \eqref{second-order-equation-W}, there exist two constants $C_{1}$ and $C_{2}$ such that
\begin{equation}\label{eq-higher-approximation-ODE-1}
W(\zeta,\xi)=[C_{1}+r_{1}(\zeta,\xi)]\Ai_{n}(\zeta,\xi)+[C_{2}+r_{2}(\zeta,\xi)]\Bi_{n}(\zeta,\xi),
\end{equation}
where
\begin{equation}
\Ai_{n}(\zeta,\xi)=\left(\frac{d\hat{\zeta}}{d\zeta}\right)^{-\frac{1}{2}}\Ai\left(\xi^{\frac{2}{3}}\hat{\zeta}\right)\quad \text{ and }\quad \Bi_{n}(\zeta,\xi)=\left(\frac{d\hat{\zeta}}{d\zeta}\right)^{-\frac{1}{2}}\Bi\left(\xi^{\frac{2}{3}}\hat{\zeta}\right),
\end{equation}
and $r_{1,2}(\zeta,\xi)=\mathcal{O}(\xi^{-2n})$
as $\xi\to+\infty$ uniformly for all $\zeta\in\mathbb{D}\cap\mathbb{S}$.
\end{lemma}

\begin{remark}\label{rem-b=0-1}
Combining \eqref{eq-def-An}, \eqref{eq-An-coeff-compare} and the analyticity of $a_{s}(\zeta)$, we see that $a_{s}(\zeta)$ are uniquely determined in a recursive manner. Specifically
\begin{equation}\label{eq-explicit-representation-a-s}
\begin{split}
a_{1}(\zeta)&=\frac{1}{2\zeta^{\frac{1}{2}}}\int_{0}^{\zeta}\frac{\varphi(t)}{t^{\frac{1}{2}}}dt,\\
a_{2}(\zeta)&=\frac{1}{2\zeta^{\frac{1}{2}}}\int_{0}^{\zeta}\frac{\psi(t)-2a_{1}(t)a_{1}'(t)-t(a_{1}'(t))^2}{t^{\frac{1}{2}}}dt,\\
a_{s}(\zeta)&=\frac{1}{2\zeta^{\frac{1}{2}}}\int_{0}^{\zeta}\frac{F_{s}(t)}{t^{\frac{1}{2}}}dt,\qquad  s>2,
\end{split}
\end{equation}
where $F_{s}(t)$ can be expressed in terms of $a_{1}(t), a_{2}(t),\cdots,a_{s-1}(t)$.
If $B(\xi)=0$, then $g(z)=0$ which implies $a_{1}(\zeta)=0$. By induction, we can show that $a_{2s+1}(\zeta)=0$ for all $s\in\mathbb{N}$. Moreover, in this case, the estimates of $r_{1,2}(\zeta,\xi)$ can be improved as
\begin{equation}\label{eq-higher-approximation-r12-improved}
r_{1,2}(\zeta,\xi)=\mathcal{O}(\xi^{-2n-1}),\qquad \xi\to+\infty.
\end{equation}
\end{remark}

\begin{remark}
The proof of Lemma \ref{lem-higher-order-approximation-ODE-1} is provided in Sec. \ref{sec:proof-lemma}, where we show that the limits of $r_{1}(\zeta,\xi)$ and $r_{2}(\zeta,\xi)$ as $\zeta\to\infty$ with $\arg{\zeta}\sim\frac{k\pi}{3}$, $k=1,3$ exist. The four limit values may not be equal to each other, but they all have an estimate of $\mathcal{O}(\xi^{-2n-1})$ as $\xi\to+\infty$. For convenience, hereafter, we set $r_{j,k}(\xi) = \lim\limits_{\zeta \to \infty} r_{i}(\zeta,\xi)$ with $\arg{\zeta} \sim \frac{k\pi}{3}$, for $k=1,3$, and $j=1,2$.
\end{remark}

By comparing Lemma \ref{lem-higher-order-approximation-ODE-1} with \cite[Theorem 2]{APC}, one may observe that Lemma \ref{lem-higher-order-approximation-ODE-1} is a refined version of \cite[Theorem 2]{APC}, as the error bounds in \eqref{eq-higher-approximation-ODE-1} are expressed in higher order terms. It is noteworthy that the special case of Lemma \ref{lem-higher-order-approximation-ODE-1} for real $\zeta$ has been considered by Dunster~\cite{Dunster-2014}, which inspires our extension of this result to the complex plane for $\zeta \in \mathbb{D} \cap \mathbb{S}$.

\paragraph{Large pole parameter case}
\

Assume $t= p$, where $p$ is any pole of the PI solution $y(t)$.  Recall that $\phi_{1}(\lambda)$ is the first component of $\hat{\Phi}(\lambda)$ in \eqref{eq-system-hat-Phi}.
It follows from eq.~\eqref{eq-system-hat-Phi} that $\phi_{1}(\lambda)$ satisfies the following reduced triconfluent Heun equation (RTHE; see \cite[eq.~(6)]{Xia-Xu-Zhao} or \cite[p.~108]{SL})
\begin{equation}\label{Schrodinger-equation-triconfluent-Heun}
\frac{d^{2}\phi_{1}}{d\lambda^{2}}=\left[4\lambda^{3}+2p\lambda-28H\right]\phi_{1}.
\end{equation}
Let $\xi$ be a large real parameter and write
\begin{equation}\label{eq-large-pole-parameter}
p=2C_{0}\xi^{\frac{4}{5}}
\quad\text{and}\quad
H=-\frac{1}{7}\xi^{\frac{6}{5}}\left(1+\frac{\tilde{h}(\xi)}{\xi^2}\right),
\end{equation}
where $C_{0}\approx 2.004860503264124$ is defined in \cite[Lemma 2.1]{LongLiWang} and $\tilde{h}(\xi)$ is a real bounded function satisfying
\begin{equation}\label{eq-tilde-h-behave-like}
\tilde{h}(\xi)=\sum\limits_{s=0}^{n-1}\frac{\tilde{h}_{s}}{\xi^{2s}}+\mathcal{O}(\xi^{-2n})
\end{equation}
as $\xi\to+\infty$ with the coefficients $\tilde{h}_{s}$ to be determined later (see Remark \ref{remark-beta-imaginary}).

\begin{remark}\label{remark-tilde-h}
Note that we assume all the odd terms in \eqref{eq-tilde-h-behave-like} vanish. This condition ensures that the Stokes multipliers obtained in Theorem~\ref{Thm-s-large-pole} below correspond to the real tritronqu\'{e}e solution of the PI equation. A detailed proof of this result is provided in Appendix~\ref{tilde-h-odd-term-vanish}.
\end{remark}

Under the scaling transformation $\lambda=\xi^{\frac{2}{5}}z$ and $\phi_{1}(\xi^{\frac{2}{5}}z)=Y(z)$, eq.~\eqref{Schrodinger-equation-triconfluent-Heun} becomes
\begin{equation}\label{eq-Schrodinger-case-II}
\frac{d^{2}Y}{dz^{2}}=\xi^{2}\left[\tilde{f}(z)+\frac{4\tilde{h}(\xi)}{\xi^2}\right]Y,
\end{equation}
where
\begin{equation}\label{eq-def-tilde-f}
\tilde{f}(z)=4(z^3+C_{0}z+1)=4(z-\tilde{z}_{0})(z-\tilde{z}_{1})(z-\tilde{z}_{2})
\end{equation}
with $\tilde{z}_{0}<0$, and $\tilde{z}_{1}, \tilde{z}_{2}$ being a conjugate
pair. We shall assume $\Im\tilde{z}_1\geq\Im\tilde{z}_2$.
Define the transformation $\eta:=\eta(z)$ by
\begin{equation}
\label{def-eta}
\frac{2}{3}\eta^{\frac{3}{2}}=\int_{\tilde{z}_{1}}^{z}\tilde{f}(t)^{\frac{1}{2}}dt,
\end{equation}
where the branches are chosen such that $\arg(z-\tilde{z}_{j})\in(-\pi,\pi), j=0,1,2$.
Then $\eta(z)$ is a conformal mapping in a neighbourhood of $z=\tilde{z}_{1}$ and the two adjacent Stokes lines emanating from $z=\tilde{z}_{1}$ to infinity. Moreover, a direct calculation of the integral in \eqref{def-eta} yields
\begin{equation}\label{eq-asymp-relation-eta-z}
\frac{2}{3}\eta^{\frac{3}{2}}=\frac{4}{5}z^{\frac{5}{2}}+2C_{0}z^{\frac{1}{2}}+\tilde{E}_{0}+\mathcal{O}(z^{-\frac{1}{2}})
\end{equation}
as $z\to\infty$ uniformly for $\arg{z}\in\left[0,\frac{4\pi}{5}\right]$, where
\begin{equation}\label{eq-def-tilde-E0}
\tilde{E_{0}}=\int_{\tilde{z}_{1}}^{\infty e^{i\theta}}
\left[2(t^3+C_{0}t+1)^{\frac{1}{2}}-\left(2t^{\frac{3}{2}}+C_{0}t^{-\frac{1}{2}}\right)\right]dt
-\left(\frac{4}{5}\tilde{z}_{1}^{\frac{5}{2}}+2C_{0}\tilde{z}_{1}^{\frac{1}{2}}\right)
\end{equation}
with ${\theta}\in\left[0,\frac{4\pi}{5}\right]$.
Set $\tilde{W}(\eta,\xi)=\left(\frac{\tilde{f}(z)}{\eta}\right)^{\frac{1}{4}}Y$, then \eqref{second-order-equation-w} becomes
\begin{equation}\label{second-order-equation-tildeW}
\frac{d^{2}\tilde{W}}{d\eta^2}=\xi^2 \left[\eta+\frac{\tilde{\psi}(\eta)+\tilde{\varphi}(\eta,\xi)}{\xi^2}\right]\tilde{W},
\end{equation}
where
\begin{equation}\label{eq-def-tilde-psi}
\tilde{\psi}(\eta)=\frac{5}{16\eta^2}+\frac{\eta[4\tilde{f}(z)\tilde{f}''(z)-5\tilde{f}'(z)^2]}{16\tilde{f}(z)^3}
\end{equation}
and
\begin{equation}
\tilde{\varphi}(\eta,\xi)=\frac{4\eta \tilde{h}(\xi)}{\tilde{f}(z)}
\end{equation}
are analytic at $\eta=0$. Similar to Lemma~\ref{lem-higher-order-approximation-ODE-1}, we have the following result.
\begin{lemma}
\label{lem-higher-order-approximation-ODE-2}
Let $\delta$ and $\epsilon$ be fixed small positive constants and denote
\[
\tilde{\mathbb{S}}:=\{\eta\in\mathbb{C}~:~\arg{\eta}=\frac{k\pi}{3}, k=1,3\}
\]
and
\[
\tilde{\mathbb{D}}:=\left\{\eta\in\mathbb{C}~:~|\eta|\leq \delta \text{ or } \frac{\pi}{3}-\epsilon<\arg{\eta}<\pi+\epsilon\right\}.
\]
For any $n\in\mathbb{N}$, define $\hat{\eta}=\eta+\mathcal{\tilde{A}}_{n}(\eta,\xi)$ with
\begin{equation}\label{eq-tilde-An-def}
\mathcal{\tilde{A}}_{n}(\eta,\xi)=\sum\limits_{s=1}^{n}\frac{\tilde{a}_{s}(\eta)}{\xi^{2s}}
\end{equation}
satisfying
\begin{equation}\label{eq-tilde-An-coeff-compare}
\xi^2\left\{\mathcal{\tilde{A}}_{n}+(\eta+\mathcal{\tilde{A}}_{n})(2+\mathcal{\tilde{A}}'_{n})\mathcal{A}'_{n}\right\}+\frac{3(\tilde{\mathcal{A}}''_{n})^{2}-2(1+\tilde{\mathcal{A}}'_{n})\tilde{\mathcal{A}}'''_{n}}{4(1+\tilde{\mathcal{A}}'_{n})^2}=\tilde{\psi}(\eta)+\tilde{\varphi}(\eta,\xi)+\mathcal{O}(\xi^{-2n})
\end{equation}
as $\xi\to+\infty$ uniformly for all $\eta\in\mathbb{\tilde{D}}$. Assume also that each $\tilde{a}_{s}(\eta)$ is analytic at $\eta=0$. Then for all $s\in\mathbb{N}$, the limits
\begin{equation}\label{eq-def-beta-s}
\beta_{s}=\lim\limits_{\eta\to\infty}\eta^{\frac{1}{2}}\tilde{a}_{s}(\eta)
\end{equation}
exist.
Moreover, for any solution of \eqref{second-order-equation-tildeW}, there exist two constants $\tilde{C}_{1}$ and $\tilde{C}_{2}$ such that
\begin{equation}\label{eq-higher-approximation-ODE-2}
\tilde{W}(\eta,\xi)=[\tilde{C}_{1}+\tilde{r}_{1}(\eta,\xi)]\Ai_{n}(\eta,\xi)+[\tilde{C}_{2}+\tilde{r}_{2}(\eta,\xi)]\Bi_{n}(\eta,\xi),
\end{equation}
where
\begin{equation}
\Ai_{n}(\eta,\xi)=\left(\frac{d\hat{\eta}}{d\eta}\right)^{-\frac{1}{2}}\Ai\left(\xi^{\frac{2}{3}}\hat{\eta}\right)\quad \text{ and }\quad \Bi_{n}(\eta,\xi)=\left(\frac{d\hat{\eta}}{d\eta}\right)^{-\frac{1}{2}}\Bi\left(\xi^{\frac{2}{3}}\hat{\eta}\right),
\end{equation}
and $\tilde{r}_{1,2}(\eta,\xi)=\mathcal{O}(\xi^{-2n-1})$ as $\xi\to+\infty$ uniformly for all $\eta\in\mathbb{\tilde{D}}\cap\mathbb{\tilde{S}}$.
\end{lemma}

\begin{remark}
Similar to $a_{s}(\zeta)$ in Lemma \ref{lem-higher-order-approximation-ODE-1},
the coefficients $\tilde{a}_{s}(\eta)$ can be determined in a recursive manner
by \eqref{eq-tilde-An-coeff-compare} and the fact that $\tilde{a}_{s}(\eta)$ are analytic in $\tilde{\mathbb{D}}$. Specifically,
\begin{equation}\label{eq-tilde-a1-expression}
\begin{split}
\tilde{a}_{1}(\eta)&=\frac{1}{2\eta^{\frac{1}{2}}}\int_{0}^{\eta}\frac{\tilde{\psi}(t)}{t^{\frac{1}{2}}}dt+\frac{1}{2\eta^{\frac{1}{2}}}\int_{z_{1}}^{z}\frac{4\tilde{h}_{0}}{\tilde{f}(z)^{\frac{1}{2}}}dz
\end{split}
\end{equation}
and
\begin{equation}\label{eq-def-tilde-a-s-expression}
\tilde{a}_{s}(\eta)=\frac{1}{2\eta^{\frac{1}{2}}}\int_{0}^{\eta}\frac{\tilde{F}_{s}(t)}{t^{\frac{1}{2}}}dt+\frac{1}{2\eta^{\frac{1}{2}}}\int_{z_{1}}^{z}\frac{4\tilde{h}_{s-1}}{\tilde{f}(z)^{\frac{1}{2}}}dz
\end{equation}
for all $s\geq 2$, where $\tilde{F}_{s}(t)$ are polynomials of $\tilde{a}_{1}(t), \tilde{a}_{2}(t),\dots,\tilde{a}_{s-1}(t)$ and their first three order derivatives. For example, $\tilde{F}_{2}(t)=\tilde{a}_{1}'''(t)-4\tilde{a}_{1}(t)\tilde{a}_{1}'(t)-2t{\tilde{a}_{1}'}(t)^2$.
\end{remark}

\begin{remark}
We may regard eq.~\eqref{eq-Schrodinger-case-II} as a special case of eq.~\eqref{second-order-equation-w},
and so that Lemma~\ref{lem-higher-order-approximation-ODE-2} as a special case of Lemma~\ref{lem-higher-order-approximation-ODE-1}.
Similar to $r_{1,2}(\zeta)$ in Lemma~\ref{lem-higher-order-approximation-ODE-1}, the limits of $\tilde{r}_{1}(\eta,\xi)$ and $\tilde{r}_{2}(\eta,\xi)$ as $\eta\to\infty$ with $\arg{\eta}\sim\frac{k\pi}{3}$, $k=1,3$, exist.
We shall denote $\tilde{r}_{j,k}(\xi)=\lim\limits_{\eta\to\infty}\tilde{r}_{i}(\eta,\xi)$ with $\arg{\eta}\sim\frac{k\pi}{3}$, $k=1,3$, and $j=1,2$.
\end{remark}

\subsection{Full asymptotic expansion of the Stokes multipliers}

Using Lemma \ref{lem-higher-order-approximation-ODE-1} and the Stokes phenomenon of the Airy functions, we obtain the full asymptotic expansion of the Stokes multipliers corresponding to the PI solution with $y(0)=a$ and $y'(0))=b$ as $\xi=\left(\frac{b^2}{4}-a^3\right)^{\frac{5}{6}}\to+\infty$.

\begin{theorem}\label{Thm-s-large-value}
Assume that $a=A(\xi)\xi^{\frac{2}{5}}$ and $b=B(\xi)\xi^{\frac{3}{5}}$ with $\frac{B(\xi)^2}{4}-A(\xi)^3=1$.
Then the Stokes multipliers corresponding to $y(t; a, b)$ have the following asymptotic expansions:
\begin{equation}\label{eq-stokes-full-asymptotic-thm1}
\begin{aligned}
s_{0}&\sim 2i\exp\left\{-\sum\limits_{s=0}^{\infty}\frac{2\Re\alpha_{s}(A(\xi),B(\xi))}{\xi^{s-1}}\right\}\left\{\cos{\left[\sum\limits_{s=0}^{\infty}-\frac{2\Im\alpha_{s}(A(\xi),B(\xi))}{\xi^{s-1}}\right]}\right\},\\
s_{1}&=-\overline{s_{-1}}\sim i\exp\left\{\sum\limits_{s=0}^{\infty}\frac{2\alpha_{s}(A(\xi),B(\xi))}{\xi^{s-1}}\right\},\\
s_{2}&=-\overline{s_{3}}\sim -i\exp\left\{\sum\limits_{s=0}^{\infty}\frac{-4i\Im\alpha_{s}(A(\xi),B(\xi))}{\xi^{s-1}}\right\}
\end{aligned}
\end{equation}
as $\xi\rightarrow+\infty$, where
\begin{equation}
\alpha_{0}(A(\xi),B(\xi))=\frac{3}{5}\B\left(\frac{1}{2},\frac{1}{3}\right)-\frac{\sqrt{3}i}{5}\B\left(\frac{1}{2},\frac{1}{3}\right)
\end{equation}
and $\alpha_{s}(A(\xi),B(\xi))$, $s=1,2,\dots$, are given in Lemma \ref{lem-higher-order-approximation-ODE-1}.
Here, $B(\cdot,\cdot)$ stands for the beta function.
\end{theorem}

\begin{remark}\label{rem-after-Thm-1}
When $A(\xi)$ and $B(\xi)$ are both fixed, $\alpha_{s}(A(\xi),B(\xi))$ are all fixed and independent of $\xi$. In particular, when $B(\xi)=0$, we have $A(\xi)=-1$, then according to Remark \ref{rem-b=0-1}, we have
\begin{equation}
\alpha_{2s+1}(-1,0)=0 \quad \text{for all} \quad s=0,1,2,\cdots.
\end{equation}
When $A(\xi)=0$, we have $B(\xi)=\pm 2$. Hence, from \eqref{eq-An-coeff-compare}, we obtain
\begin{equation}
\begin{split}
\alpha_{1}(0,\pm 2)&=\lim\limits_{\zeta\to\infty}\frac{1}{2}\int_{0}^{\zeta}\frac{\pm \varphi(t)}{t^{\frac{1}{2}}}dt=\frac{1}{2}\int_{e^{\frac{\pi i}{3}}}^{\infty}\frac{g(z)}{f(z)^{\frac{1}{2}}}dz=\mp\int_{e^{\frac{\pi i}{3}}}^{\infty}\frac{1}{2z(z^3+1)^{\frac{1}{2}}}dz=\pm                                                                                                                                                                                                                                                                                                                                                                                                                                                                                                                                                                                                                                                                                                                                                                                                                                                                                                                                                                                                                                                                                                                                                                                                                                                                                                                            \frac{\pi i}{6}.
\end{split}
\end{equation}
However, when $A(\xi)$ or $B(\xi)$ is not fixed, we find that $\alpha_{s}(A(\xi),B(\xi))$ all depend on $\xi$. For instance, when $a\neq 0$ is fixed, $A(\xi)=a\xi^{-\frac{2}{5}}$ and $B(\xi)=\pm \left(2+a^3\xi^{-\frac{6}{5}}\right)+\mathcal{O}(\xi^{-\frac{12}{5}})$ as $\xi\to+\infty$. This implies that $\alpha_{1}(A(\xi),B(\xi))=\pm\frac{\pi i}{6}+\mathcal{O}(\xi^{-\frac{2}{5}})$ in this case. Similarly, when $b\neq 0$ is fixed, we have $B(\xi)=b\xi^{-\frac{3}{5}}$ and $A(\xi)=-1+\frac{b^2}{12}\xi^{-\frac{6}{5}}+\mathcal{O}(\xi^{-\frac{12}{5}})$ as $\xi\to+\infty$. This implies that $\alpha_{1}(A(\xi),B(\xi))=\mathcal{O}(\xi^{-\frac{3}{5}})$ in this case.
\end{remark}

Next, let us consider the PI solution with large pole parameters in \eqref{eq-large-pole-parameter}.
Using Lemma \ref{lem-higher-order-approximation-ODE-2} and the Stokes phenomenon of the Airy functions, we obtain the full asymptotic expansion of the Stokes multipliers corresponding to such PI solution
as $\xi\to+\infty$.

\begin{theorem}\label{Thm-s-large-pole}
Let $C_{0}$ be the constant given in \cite[Lemma 1.2]{LongLiWang} and let $(p,H)$ be parametrized as in \eqref{eq-large-pole-parameter}. The full asymptotic expansion of the Stokes multipliers, corresponding to the real tritronqu\'{e}e solution of PI equation, are given by
\begin{equation}
\begin{split}
\label{eq-stokes-full-asymptotic-thm2}
s_{0}&\sim i\left\{-2\cos{\left[\sum\limits_{s=0}^{\infty}-\frac{2i\beta_{s}}{\xi^{2s-1}}\right]}-1\right\},\\
s_{1}&\sim-\overline{s_{-1}}\sim -i\exp\left\{\sum\limits_{s=0}^{\infty}\frac{2\beta_{s}}{\xi^{2s-1}}\right\},\\
s_{2}&\sim-\overline{s_{-2}}\sim -2i\exp\left\{-\sum\limits_{s=0}^{\infty}\frac{3\beta_{s}}{\xi^{2s-1}}\right\}\cos\left\{\sum\limits_{s=0}^{\infty}\frac{i\beta_{s}}{\xi^{2s-1}}\right\}
\end{split}
\end{equation}
as $\xi\to+\infty$, where $\beta_{s}, s=1,2,\dots$, are given in Lemma~\ref{lem-higher-order-approximation-ODE-2}. In particular, we have $\beta_{0}=\tilde{E}_{0}$.

\end{theorem}

\begin{remark}\label{remark-beta-imaginary}
Note that the real tritronqu\'{e}e solution has the Stokes multipliers $s_{1}=s_{-1}=i,$ $s_{2}=s_{-2}=0,$ $s_{0}=i$. Consequently, by Theorem~\ref{Thm-s-large-pole}, we must have \begin{equation}\label{eq-Re-betas=0}
\Re{\beta_{s}}=0,\qquad\forall s\in\mathbb{N}.	
\end{equation}
Indeed, according to \cite[eqs. (2.2), (2.6) and (2.7)]{LongLiWang}, it follows that $\Re \beta_{0}=\Re\tilde{E}_{0}=0$ holds. Using \eqref{eq-Re-betas=0}, we can recursively determine $\tilde{h}_{s}$ for all $s\in\mathbb{N}$.
In particular, from the definition of $\beta_{1}$ in \eqref{eq-def-beta-s} and \eqref{eq-tilde-a1-expression}, the coefficient $\tilde{h}_{0}$ is determined by
\begin{equation}
\Re \beta_{1}=\Re{d_{1}}+4\tilde{h}_{0}\Re{d_{0}}=0,
\end{equation}
where
\begin{equation}\label{eq-explicit-integral-d0-d1}
d_{1}=\frac{1}{2}\int_{z_{1}}^{\infty e^{i\theta}}\frac{\tilde{\psi}(\eta(z))\tilde{f}(z)^{\frac{1}{2}}}{\eta(z)}dz \quad\text{and}\quad d_{0}=\frac{1}{2}\int_{z_{1}}^{\infty e^{i\theta}}\frac{1}{\tilde{f}(z)^{\frac{1}{2}}}dz,
\end{equation}
with $\theta\in[0,\frac{4\pi}{5}]$.
This further implies
\begin{equation}\label{eq-explicit-tilde-h0-beta1}
\tilde{h}_{0}=-\frac{\Re d_{1}}{4\Re d_{0}}\quad { and }\quad \beta_{1}=d_{1}-\frac{\Re d_{1}}{\Re d_{0}}d_{0}.
\end{equation}
\end{remark}

\subsection{Full asymptotic expansions of the nonlinear eigenvalues and pole parameters}

An asymptotic expansion of the nonlinear eigenvalues of PI is proposed in \cite{Bender-Komijani-Wang}
based on numerical simulations.
Now, we can show that such a full asymptotic expansion is indeed valid
as a direct consequence of Theorem~\ref{Thm-s-large-value}.

\begin{corollary}\label{cor-1}
For $a=0$, {\it i.e.} $A(\xi)=0, B(\xi)=\pm 2$, there exist two sequences $\{b_{n}^{+}\}$ and $\{b_{n}^{-}\}$ such that the PI solutions $y(t; a, b_{n}^{\pm})$ are separatrix solutions, and the following asymptotic expansions hold:
\begin{equation}\label{eq-full-expansion-bn+}
b_{n}^{+}\sim 2\left[\frac{\sqrt{3\pi}\Gamma\left(\frac{11}{6}\right)}{\Gamma\left(\frac{1}{3}\right)}\right]^{\frac{3}{5}}\left(n-\frac{1}{6}\right)^{\frac{3}{5}}\left[1+\sum\limits_{s=2}^{\infty}\frac{B_{s}^{+}}{(n-\frac{1}{6})^s}\right]
\end{equation}
and
\begin{equation}\label{eq-full-expansion-bn-}
b_{n}^{-}\sim -2\left[\frac{\sqrt{3\pi}\Gamma\left(\frac{11}{6}\right)}{\Gamma\left(\frac{1}{3}\right)}\right]^{\frac{3}{5}}\left(n-\frac{5}{6}\right)^{\frac{3}{5}}\left[1+\sum\limits_{s=2}^{\infty}\frac{B_{s}^{-}}{(n-\frac{5}{6})^s}\right]
\end{equation}
as $n\to\infty$, where $B_{s}^{\pm}, s=2,3,\dots$, can be expressed in terms of $\alpha_{k}(0,\pm 2), k=0,1,2,\cdots,s$, where $\alpha_{s}(A(\xi),B(\xi))$ is defined in \eqref{eq-def-alpha-s}.  In particular, we have
\begin{equation}\label{eq-B2pm-explicit-by-alpha2}
B_{2}^{\pm}=\frac{12\sqrt{3}}{25\pi^2}\B\left(\frac{1}{2},\frac{1}{3}\right)\Im \alpha_{2}(0,\pm 2),
\end{equation}
and
\begin{equation}\label{eq-alpha2-integral}
\alpha_{2}(0,\pm 2)=\int_{0}^{\infty e^{i\theta}}\frac{\psi(t)-2a_{1}(t)a_{1}'(t)-ta_{1}'(t)^2}{t^{\frac{1}{2}}}dt, \qquad \theta\in[0,\frac{4\pi}{3}],
\end{equation}
where $\psi(\zeta)$ and $a_{1}(\zeta)$ are defined in \eqref{eq-def-varphi-psi} and \eqref{eq-explicit-representation-a-s} respectively.

\end{corollary}

\begin{proof}
When $a=0$, we have $B(\xi)=\pm 2$. Regard the Stokes multipliers $s_{k}$ as functions of $\xi$. According to \eqref{eq-stokes-full-asymptotic-thm1}, it is evident that there exist two sequences, $\{\xi_{n}^{+}\}$ and $\{\xi_{n}^{-}\}$, such that $s_{0}(\xi_{n}^{\pm})=0$. This implies that the solutions satisfying $y(0)=0$ and $y'(0)=\pm 2\xi_{n}^{\frac{3}{5}}$ are separatrix solutions. Furthermore, we have
\begin{equation}
-2\Im \alpha_{0}(0,\pm 2)\xi_{n}^{\pm}-2\Im \alpha_{1}(0,\pm 2)-\sum\limits_{s=1}^{\infty}\frac{2\Im \alpha_{s+1}(0,\pm 2)}{(\xi_{n}^{\pm})^{s}}\sim\left(n-\frac{1}{2}\right)\pi
\end{equation}
as $n\to\infty$. This, in turn, implies that
\begin{equation}\label{eq-asym-equation-xi-n}
\xi_{n}^{\pm}\sim-\frac{\left(n-\frac{1}{2}\pm\frac{1}{3}\right)\pi}{2\Im \alpha_{0}(0,\pm 2)}-\sum\limits_{s=1}^{\infty}\frac{2\Im \alpha_{s+1}(0,\pm 2)}{2\Im \alpha_{0}(0,\pm 2)(\xi_{n}^{\pm})^{s}}
\end{equation}
as $n\to\infty$. Noting that $b=B(\xi)\xi^{\frac{3}{5}}=\pm 2\xi^{\frac{3}{5}}$, it follows that $b_{n}^{\pm}=\pm 2(\xi^{\pm}_{n})^{\frac{3}{5}}$.  We can then obtain \eqref{eq-full-expansion-bn+} and \eqref{eq-full-expansion-bn-}, where the coefficients $B_{s}^{\pm}$ can be expressed in term of $\alpha_{k}(0,\pm 2)$ for $k=0,1,\cdots, s$. Specifically, we have
\begin{equation}
B_{2}^{\pm}=-\frac{12\Im \alpha_{0}(0,\pm 2)\Im\alpha_{2}(0,\pm 2)}{5\pi^2}=\frac{12\sqrt{3}}{25\pi^2}B\left(\frac{1}{2},\frac{1}{3}\right)\Im\alpha_{2}(0,\pm 2).
\end{equation}
Moreover, by combining \eqref{eq-def-alpha-s} and \eqref{eq-explicit-representation-a-s}, and taking the limit $\zeta\to\infty$ in $\mathbb{D}$, we obtain
\begin{equation*}
\alpha_{2}(0,\pm 2)=\lim\limits_{\zeta\to\infty}\zeta^{\frac{1}{2}}a_{2}(\zeta)=\int_{0}^{\infty e^{i\theta}}\frac{\psi(t)-2a_{1}(t)a_{1}'(t)-ta_{1}'(t)^2}{t^{\frac{1}{2}}}dt, \quad \theta\in[0,\frac{4\pi}{3}];
\end{equation*}
completing the proof of Corollary~\ref{cor-1}.
\end{proof}

\begin{remark}
By solving the asymptotic expansion of $\xi_{n}^{\pm}$ from \eqref{eq-asym-equation-xi-n}, and noting that $b_{n}^{\pm} = \pm 2(\xi_{n}^{\pm})^{\frac{3}{5}}$, we can express all the coefficients $B_{s}^{\pm}$ in terms of $\alpha_{k}(0,\pm 2)$ for $k = 0, 1, \ldots, s$. Furthermore, using \eqref{eq-def-alpha-s} and \eqref{eq-explicit-representation-a-s}, it follows that $\alpha_{k}(0,\pm 2)$ can be represented through multiple integrals. However, evaluating these integrals explicitly appears challenging.  Therefore, we numerically compute the integral in \eqref{eq-alpha2-integral} and substitute the resulting value into \eqref{eq-B2pm-explicit-by-alpha2}, obtaining $B_{2}^{+} = B_{2}^{-} \approx -0.005516$. This result aligns well with the numerical findings reported in \cite{Bender-Komijani-Wang}.
\end{remark}

\begin{remark}
In Corollary \ref{cor-1}, if $a\neq 0$ is fixed, according to the analysis in Remark \ref{rem-after-Thm-1}, we can also obtain the following asymptotic expansion of $(b_{n}^{\pm})^2-4a^3$ as $n\to+\infty$
\begin{equation}\label{eq-full-expansion-bn-fixed-a}
\left((b_{n}^{\pm})^2-4a^3\right)^{\frac{1}{2}}\sim \pm 2\left[\frac{\sqrt{3\pi}\Gamma\left(\frac{11}{6}\right)}{\Gamma\left(\frac{1}{3}\right)}\right]^{\frac{3}{5}}\left(n-\gamma_{n}^{\pm}\right)^{\frac{3}{5}}\left[1+\sum\limits_{s=2}^{\infty}\frac{B_{s}^{\pm}(n)}{(n-\gamma_{n}^{\pm})^s}\right],
\end{equation}
which is similar to \eqref{eq-full-expansion-bn+} and \eqref{eq-full-expansion-bn-}. Here, $\gamma_{n}^{\pm}=-\frac{1}{2}\pm\frac{1}{3}+\mathcal{O}(n^{-\frac{2}{5}})$ as $n\to+\infty$ and $B_{s}^{\pm}(n)$ are all bounded but depend on $n$.
\end{remark}

When $b=0$, then $A(\xi)=-1$ and $B(\xi)=0$. Noting that $\alpha_{2s+1}(-1,0)=0$, we have the following result.
\begin{corollary}\label{cor-2}
For $b=0$, {\it i.e.} $A(\xi)=-1, B(\xi)=0$, there exists a sequence $\{a_{n}\}$ such that the PI solutions $y(t; a_{n}, b)$ are separatrix solutions and
\begin{equation}\label{eq-asym-an-b=0}
a_{n}\sim -\left[\frac{\sqrt{3\pi}\Gamma\left(\frac{11}{6}\right)}{\Gamma\left(\frac{1}{3}\right)}\right]^{\frac{2}{5}}\left(n-\frac{1}{2}\right)^{\frac{2}{5}}\left[1+\sum\limits_{s=1}^{\infty}\frac{A_{s}}{(n-\frac{1}{2})^{2s}}\right]
\end{equation}
as $n\to\infty$, where $A_{s}$ can be expressed in terms of $\alpha_{2k}(-1,0), k=1,2,\dots,s$. In particular, we have
\begin{equation}\label{eq-A1-explicit-by-alpha2}
A_{1}=\frac{8\sqrt{3}}{25\pi^2}\B\left(\frac{1}{2},\frac{1}{3}\right)\Im \alpha_{2}(-1,0),
\end{equation}
and
\begin{equation}\label{eq-alpha2-integral-A1}
\alpha_{2}(-1,0)=\int_{0}^{\infty e^{i\theta}}\frac{\psi(t)}{t^{\frac{1}{2}}}dt, \qquad \theta\in [0,\frac{4\pi}{3}]
\end{equation}
where $\psi(\zeta)$ is defined in \eqref{eq-def-varphi-psi} and $\mathrm{B}(\cdot,\cdot)$ is the beta function.

\end{corollary}

\begin{remark}
Through numerical simulations, we find that $A_{1}\approx -0.009651797894657,$ which is consistent with the numerical results reported in \cite{Bender-Komijani-Wang}
\end{remark}

\begin{remark}
In Corollary \ref{cor-2}, if $b\neq 0$ is fixed, according to the analysis in Remark \ref{rem-after-Thm-1}, we can also obtain the following asymptotic expansion of $a_{n}$ as $n\to\infty$
\begin{equation}
\left(a_{n}^3-\frac{b^2}{4}\right)^{\frac{1}{3}} \sim -\left[\frac{\sqrt{3\pi}\Gamma\left(\frac{11}{6}\right)}{\Gamma\left(\frac{1}{3}\right)}\right]^{\frac{2}{5}}\left(n-\mu_{n}\right)^{\frac{2}{5}}\left[1+\sum\limits_{s=1}^{\infty}\frac{A_{s}(n)}{(n-\mu_{n})^{2s}}\right],
\end{equation}
which is an analogue to \eqref{eq-asym-an-b=0}. Here, $\mu_{n}=\frac{1}{2}+\mathcal{O}(n^{-\frac{3}{5}})$ as $n\to\infty$ and $A_{s}(n)$ are all bounded but depend on $n$.
\end{remark}

\

Let us turn to the real tritronqu\'{e}e solution of PI.
We can obtain the full asymptotic expansions for the pole parameters $(p_{n}, H_{n})$ as $n\to\infty$,
where $p_n$ denotes the $n$-th pole of this solution.
It is a direct consequence of the combination of Theorem \ref{Thm-s-large-pole} and the fact that the Stokes multipliers corresponding to PI's real tritronqu\'{e}e solution are $s_{1}=s_{-1}=i, s_{2}=s_{-2}=0, s_{0}=i$.

\begin{corollary}
Let $(p_{n}, H_{n})$ be the pole parameters of the $n$-th real pole of PI's real tritronqu\'{e}e solution, then
\begin{equation}\label{eq-asym-expansion-pn-Hn}
\begin{split}
p_{n}&\sim 2C_{0}\left(\frac{4n-2}{\kappa^{2}(C_{0})}\right)^{\frac{4}{5}}\left[1+\sum\limits_{s=1}^{\infty}\frac{\rho_{s}}{\left(n-\frac{1}{2}\right)^{2s}}\right],\\
H_{n}&\sim -\frac{1}{7}\left(\frac{4n-2}{\kappa^{2}(C_{0})}\right)^{\frac{6}{5}}\left[1+\sum\limits_{s=1}^{\infty}\frac{\mathcal{H}_{s}}{\left(n-\frac{1}{2}\right)^{2s}}\right]
\end{split}
\end{equation}
as $n\to\infty$, where $C_{0}$ and $\kappa^{2}(C_{0})$ are given in \cite[Lemma 2.1]{LongLiWang}. Moreover, the coefficients $\rho_{s}$ and $\mathcal{H}_{s}$ can be expressed in terms of $\beta_{0},\beta_{1},\dots, \beta_{s}$ and $\tilde{h}_{0},\tilde{h}_{1},\cdots,\tilde{h}_{s-1}$.
In particular, we have
\begin{equation}\label{eq-explicit-rho1-H1}
\rho_{1}=\frac{4\beta_{0}\beta_{1}}{5\pi^2}
\quad\text{and}\quad 
\mathcal{H}_{1}=\frac{6\beta_{0}\beta_{1}}{5\pi^2}+\tilde{h}_{0}.
\end{equation}
\end{corollary}

\begin{proof}
From \eqref{eq-stokes-full-asymptotic-thm2}, we know that there exists a sequence $\{\xi_{n}\}$ such that $s_{2}=s_{-2}(\xi_{n})=(\xi_{n})=0, s_{0}(\xi_{n})=s_{1}(\xi_{n})=s_{-1}(\xi_{n})=i$ and
\begin{equation}\label{eq-asym-xi-n-pole-case}
i\beta_{0}\xi_{n}+\frac{i\beta_{1}}{\xi_{n}}+\frac{i\beta_{2}}{\xi_{n}^{3}}+\cdots  \sim (n-\frac{1}{2})\pi
\end{equation}
as $n\to\infty$. To conclude that the right hand side of the above approximation is $(n-\frac{1}{2})\pi$ but not  $-(n-\frac{1}{2})\pi$, we should observe that $\beta_{0}=\tilde{E}_{0}$ and compute $\tilde{E_{0}}$ numerically to find that $i\beta_{0}>0$. One can obtain the same conclusion by comparing the definition of $\tilde{E}_{0}$ with \cite[eqs. (2.2), (2.6) and (2.7)]{LongLiWang}. Actually, one can find that $\tilde{E}_{0}=\frac{1}{4}\kappa^{2}(C_{0})\pi i$, where $\kappa^{2}(C)$ is defined in \cite[Lemma 2.1]{LongLiWang}. Solving the asymptotic expansion of $\xi_{n}$ as $n\to\infty$ from \eqref{eq-asym-xi-n-pole-case} and inserting it into
\begin{equation}\label{eq-pn-Hn-by-xi-n}
p_{n}=2C_{0}\xi_{n}^{\frac{4}{5}} \quad\text{and}\quad H_{n}=-\frac{1}{7}\xi_{n}^{\frac{6}{5}}\left(1+\tilde{h}(\xi_{n})\xi_{n}^{-2}\right),
\end{equation}
we obtain \eqref{eq-asym-expansion-pn-Hn} and conclude that $\rho_{s}$ and $\mathcal{H}_{s}$ can be expressed in terms of $\beta_{k}, k=0,1,\cdots,s$ and $\tilde{h}_{k}, k=0,1,\cdots,s-1$. Precisely, form \eqref{eq-asym-xi-n-pole-case}, we have
\begin{equation}
\xi_{n}\sim \frac{(n-\frac{1}{2})\pi}{i\beta_{0}}\left[1+\frac{\beta_{1}\beta_{0}}{(n-\frac{1}{2})\pi i\beta_{0}\xi_{n}}+\frac{\beta_{0}\beta_{2}}{(n-\frac{1}{2})\pi i\beta_{0}\xi_{n}^{3}}+\cdots\right].
\end{equation}
Substitute it into \eqref{eq-pn-Hn-by-xi-n} and comparing the coefficients of the sub-leading term in \eqref{eq-asym-expansion-pn-Hn} and \eqref{eq-pn-Hn-by-xi-n}, we get \eqref{eq-explicit-rho1-H1}.
\end{proof}

\begin{remark}
By numerically evaluating the integrals in \eqref{eq-def-tilde-E0} and \eqref{eq-explicit-integral-d0-d1}, and substituting the resulting values into \eqref{eq-explicit-tilde-h0-beta1}, followed by substitution into \eqref{eq-explicit-rho1-H1}, we obtain $\rho_{1}\approx 0.00451478\cdots$ and $\mathcal{H}_{1}\approx -0.00881494\cdots$, which are consistent with the numerical results reported in \cite{LongLiWang}.
\end{remark}

\section{Uniform asymptotics and proof of Lemma \ref{lem-higher-order-approximation-ODE-1}}
\label{sec:proof-lemma}
To prove Lemma \ref{lem-higher-order-approximation-ODE-1}, we first use induction to show that
\begin{equation}\label{bounded-appro-a}
a_{s}(\zeta)=\mathcal{O}(\zeta^{-\frac{1}{2}}),\quad a_{s}^{(k)}(\zeta)=\mathcal{O}(\zeta^{-\frac{3}{2}}), \quad k=1,2,\dots,
\end{equation}
as $\zeta\to\infty$ in $\mathbb{D}$ for all $s\in\mathbb{N}$.

Let $s=1,2$. It follows from \eqref{eq-asymp-relation-zeta-z} and \eqref{eq-def-varphi-psi} that
\begin{equation}
\varphi(\zeta)=\mathcal{O}(\zeta^{-\frac{7}{5}}), \quad \psi(\zeta)=\mathcal{O}(\zeta^{-2}),
\qquad\text{as}~\zeta\to\infty,
\end{equation}
which, together with the explicit representations of $a_{1}(\zeta)$ and $a_{2}(\zeta)$ in \eqref{eq-explicit-representation-a-s}, leads to \eqref{bounded-appro-a}.

Suppose \eqref{bounded-appro-a} holds for $s=1,2,\dots, m-1$ with $m>2$.
It then follows from \eqref{eq-An-coeff-compare} and the analyticity of $a_{m}(\zeta)$ in $\mathbb{D}$ that
\begin{equation}\label{eq-am-integral-representation}
a_{m}(\zeta)=\frac{1}{2\zeta^{\frac{1}{2}}}\int_{0}^{\zeta}\frac{F_{m}(t)}{t^{\frac{1}{2}}}dt,
\end{equation}
where $F_{m}(\zeta)$ is a polynomial of $a_{s}'(\zeta), a_{s}''(\zeta), a_{s}'''(\zeta), s=1,2,\cdots, m-1$. Therefore, $F_{m}(\zeta)$ and its derivatives all have an estimate of $\mathcal{O}(\zeta^{-\frac{3}{2}})$ as $\zeta\to\infty$. Hence, the integral in \eqref{eq-am-integral-representation} is convergent when $\zeta\to\infty$, and then $a_{m}(\zeta)=\mathcal{O}(\zeta
^{-\frac{1}{2}})$. According to \eqref{eq-An-coeff-compare} [or \eqref{eq-am-integral-representation}],
we also have
\begin{equation}\label{eq-relation-a-F}
a_{m}(\zeta)+2\zeta a_{m}'(\zeta)=F_{s}(\zeta)=\mathcal{O}(\zeta^{-\frac{3}{2}}),
\end{equation}
which implies that the derivative of $a_{m}(\zeta)$ of any order has the estimate of $\mathcal{O}(\zeta^{-\frac{3}{2}})$ as $\zeta\to\infty$.
Therefore, \eqref{bounded-appro-a} holds for any $s=1,2,\dots$ by induction.

It follows from \eqref{eq-def-An} and \eqref{bounded-appro-a} that $\zeta^{\frac{1}{2}}\mathcal{A}_{n}(\zeta,\xi)=\mathcal{O}(\xi^{-1})$ as $\xi\to+\infty$ uniformly for all $\zeta\in\mathbb{D}$. Hence,
\begin{equation}
e^{\xi\hat{\zeta}^{\frac{3}{2}}}=\exp\{\xi\zeta^{\frac{3}{2}}\left(1+\zeta^{-1}\mathcal{A}_{n}(\zeta,\xi)\right)^{\frac{3}{2}}\}=\exp\{\xi\zeta^{\frac{3}{2}}\}\cdot \mathcal{O}(\xi\zeta^{\frac{1}{2}}\mathcal{A}_{n}(\zeta,\xi))
\end{equation}
as $\xi\to+\infty$, where $\zeta\in\mathbb{D}$.
Note that $\zeta^{\frac{3}{2}}$ is purely imaginary when $\zeta\in\mathbb{D}\cap\mathbb{S}$.
Therefore, $e^{\xi\hat{\zeta}^{\frac{3}{2}}}=\mathcal{O}(1)$ as $\xi\to+\infty$ uniformly for $\zeta\in\mathbb{D}\cap\mathbb{S}$.

\

Now, we are in a position to prove Lemma~\ref{lem-higher-order-approximation-ODE-1}.

\textbf{Proof of Lemma \ref{lem-higher-order-approximation-ODE-1}}
It is clear from \eqref{bounded-appro-a} and \eqref{eq-am-integral-representation} that the limits  $\lim\limits_{\zeta\to\infty}\zeta^{\frac{1}{2}}a_{s}(\zeta)$ exist for all $s=1,2,\cdots$. It is also evident that $\Ai_{n}(\zeta,\xi)$ and $\Bi_{n}(\zeta,\xi)$ satisfy the equation
\begin{equation}
\frac{d^2w}{d\zeta^2}=\xi^2\left[\zeta+Q_{n}(\zeta,\xi)\right]w,
\end{equation}
where
\begin{equation}
Q_{n}(\zeta,\xi)=\left\{\mathcal{A}_{n}+(\zeta+\mathcal{A}_{n})(2+\mathcal{A}'_{n})\mathcal{A}'_{n}\right\}+\frac{3{\mathcal{A}''_{n}}^{2}-2(1+\mathcal{A}'_{n})\mathcal{A}'''_{n}}{4\xi^2(1+\mathcal{A}'_{n})^2}.
\end{equation}
It follows from \eqref{eq-def-An} and \eqref{eq-relation-a-F} that $\mathcal{A}_{n}+2\zeta\mathcal{A}_{n}'=\mathcal{O}(\zeta^{-\frac{3}{2}})$,
which, together with \eqref{bounded-appro-a}, imply that
$Q_{n}(\zeta,\xi)=\mathcal{O}(\zeta^{-\frac{3}{2}})$ as $\zeta\to\infty$ in $\mathbb{D}$.
For any solution $W(\zeta,\xi)$ of \eqref{second-order-equation-W}, we have
\begin{equation}
\frac{d^2W}{d\zeta^2}=\xi^2\left[\zeta+Q_{n}(\zeta,\xi)\right]W+\left[\xi\varphi(\zeta)+\psi(\zeta)-\xi^{2}Q_{n}(\zeta,\xi)\right]W.
\end{equation}
Then there exist two constants $C_{1}$ and $C_{2}$ such that for all $\zeta\in\mathbb{D}\cap\mathbb{S}$
\begin{equation}\label{eq-constant-variation-formula}
\begin{split}
W(\zeta,\xi)=&C_{1}\Ai_{n}(\zeta,\xi)+C_{2}\Bi_{n}(\zeta,\xi)\\
&\quad+\int_{0}^{\zeta}\frac{\Ai_{n}(t,\xi)\Bi_{n}(\zeta,\xi)-\Ai_{n}(\zeta,\xi)\Bi_{n}(t,\xi)}{\mathcal{W}(\Ai_{n}(t,\xi),
\Bi_{n}(t,\xi))}W(t,\xi)R(t,\xi)dt
\end{split}
\end{equation}
where
\begin{equation}\label{eq-R-asym-zeta-infty}
R(\zeta,\xi)=\left[\xi\varphi(\zeta)+\psi(\zeta)-\xi^{2}Q_{n}(\zeta,\xi)\right]=\mathcal{O}(\zeta^{-\frac{3}{2}}) ,\quad \zeta\to\infty,~ \zeta\in\mathbb{D}
\end{equation}
and
\begin{equation}\label{eq-R-asym-xi-infty}
R(\zeta,\xi)=\left[\xi\varphi(\zeta)+\psi(\zeta)-\xi^{2}Q_{n}(\zeta,\xi)\right]=\mathcal{O}(\xi^{-2n+1}) ,\quad \xi\to+\infty
\end{equation}
uniformly for all $\zeta\in\mathbb{D}\cap\mathbb{S}$. Applying the iterative argument used in the proof of \cite[Theorem 2]{APC}, we see that $W(\zeta,\xi)$ is bounded on $\mathbb{D}\cap\mathbb{S}$. A combination of this fact and \eqref{eq-R-asym-zeta-infty} yields that the last term in \eqref{eq-constant-variation-formula} is integrable.
Using the Wronskian for Airy functions (\cite[eq.~(9.2.7)]{NIST-handbook}), we have
\begin{equation}\label{eq-Wronskian-approx}
\mathcal{W}((\Ai_{n}(t,\xi),
\Bi_{n}(t,\xi)))\sim \frac{1}{\pi}\xi^{\frac{2}{3}},\qquad \xi\to+\infty,
\end{equation}
we obtain \eqref{eq-higher-approximation-ODE-1}, and
\begin{equation}\label{eq-approx-r12-first}
r_{1,2}(\zeta,\xi)=\mathcal{O}(\xi^{-2n+\frac{1}{3}})
\end{equation}
as $\xi\to+\infty$ uniformly for all $\zeta\in\mathbb{D}\cap\mathbb{S}$.

To get the desired approximations of $r_{1,2}(\zeta,\xi)$ as stated in Lemma \ref{lem-higher-order-approximation-ODE-1}, we need the following approximation of $W(\zeta,\xi)$
\begin{equation}\label{eq-bound-W-by-Airy}
W(\zeta,\xi)\leq M(|C_{1}||\Ai_{n}(\zeta,\xi)|+|C_{2}||\Bi_{n}(\zeta,\xi)|)
\end{equation}
which can be obtained as a direct consequence of \eqref{eq-higher-approximation-ODE-1} and
\eqref{eq-approx-r12-first}.
Let $\zeta^{\star}$ be a point on the integration contour in \eqref{eq-constant-variation-formula} satisfying $|\zeta^{\star}|\sim \xi^{-\frac{1}{3}}$.
A comparison of  \eqref{eq-constant-variation-formula} and \eqref{eq-higher-approximation-ODE-1} leads to
\begin{equation}
\begin{split}
r_{1}(\zeta,\xi)&=-\int_{0}^{\zeta}\frac{\Bi_{n}(t,\xi)W(t,\xi)R(t,\xi)}{\mathcal{W}(\Ai_{n}(t,\xi),\Bi_{n}(t,\xi))}dt\\
&=-\int_{0}^{\zeta^{\star}}\frac{\Bi_{n}(t,\xi)W(t,\xi)R(t,\xi)}{\mathcal{W}(\Ai_{n}(t,\xi),\Bi_{n}(t,\xi))}dt-\int_{\zeta^{\star}}^{\zeta}\frac{\Bi_{n}(t,\xi)W(t,\xi)R(t,\xi)}{\mathcal{W}(\Ai_{n}(t,\xi),\Bi_{n}(t,\xi))}dt\\
&=I_{1}+I_{2}.
\end{split}
\end{equation}
The integrand of $I_{1}$ is bounded by $\mathcal{O}(\xi^{-2n+\frac{1}{3}})$, then
\begin{equation}\label{eq-approx-I1}
I_{1}=\mathcal{O}(\xi^{-2n})\quad \text{as} \quad \xi\to+\infty.
\end{equation}
To estimate $I_{2}$, we use \eqref{eq-R-asym-xi-infty}, \eqref{eq-Wronskian-approx} and \eqref{eq-bound-W-by-Airy}  to get
\begin{equation}\label{eq-approx-I2}
\left|I_{2}\right|\leq \frac{4\pi M^2}{\xi} \int_{\zeta^{\star}}^{\zeta}\left|t^{-\frac{1}{2}}R(t,\xi)\right| dt\leq \tilde{M}\xi^{-2n}
\end{equation}
where $\tilde{M}>0$ is a constant. This fact, together with \eqref{eq-R-asym-zeta-infty}, \eqref{eq-approx-I1} and \eqref{eq-approx-I2}, lead to $r_{1}(\zeta,\xi)=\mathcal{
O}(\xi^{-2n})$ as $\xi\to+\infty$. In a similar manner, we can show that $r_{2}(\zeta,\xi)$ has the same estimate as $r_{1}(\zeta,\xi)$; completing the proof of Lemma~\ref{lem-higher-order-approximation-ODE-1}.

\section{Asymptotic matching and proof of Theorems \ref{Thm-s-large-value} and \ref{Thm-s-large-pole}}
\label{sec:proof-theorem}

In this section, we derive full asymptotic expansions of the Stokes multipliers $s_{k}$'s in case (ii) [\textit{resp.} case (iii)], which is done by matching the uniform asymptotic approximation of $W(\zeta,\xi)$ stated in Lemma \ref{lem-higher-order-approximation-ODE-1} [\textit{resp.} Lemma~\ref{lem-higher-order-approximation-ODE-2}] and the asymptotic expansions of the canonical solutions $\hat{\Phi}_{k}$ in \eqref{eq-canonical-solutions-hat-Phi}. The approach is similar to that used in \cite{LongLi, LongLiWang}, and hence we only state some crucial steps, referring to these two articles for details.

\

\textbf{Proof of Theorem \ref{Thm-s-large-value}}

Since the Stokes multipliers $s_{k}$'s depend on $\xi$,
hence we will only take the limit $\lambda\to\infty$
when we do the asymptotic matching between the Airy functions and $\Phi_{k}$'s.
Recall the asymptotics for the Airy functions for large arguments (see~\cite[eqs. (9.2.12), (9.7.5), (9.2.10)]{NIST-handbook})
\begin{eqnarray}\label{eq-asym-Ai}
\begin{cases}
\Ai(Z)\sim
\frac{1}{2\sqrt{\pi}}Z^{-\frac{1}{4}}e^{-\frac{2}{3}Z^{\frac{3}{2}}},
\quad & \arg{Z}\in(-\pi,\pi),\\
\Ai(Z)\sim
\frac{1}{2\sqrt{\pi}}Z^{-\frac{1}{4}}e^{-\frac{2}{3}Z^{\frac{3}{2}}}
+\frac{i}{2\sqrt{\pi}}Z^{-\frac{1}{4}}e^{\frac{2}{3}Z^{\frac{3}{2}}},
\quad & \arg{Z}\in\left(\frac{\pi}{3},\frac{5\pi}{3}\right),\\
\Ai(Z)\sim
\frac{1}{2\sqrt{\pi}}Z^{-\frac{1}{4}}e^{-\frac{2}{3}Z^{\frac{3}{2}}}
-\frac{i}{2\sqrt{\pi}}Z^{-\frac{1}{4}}e^{\frac{2}{3}Z^{\frac{3}{2}}},
\quad &\arg{Z}\in\left(-\frac{5\pi}{3},-\frac{\pi}{3}\right)
\end{cases}
\end{eqnarray}
and
\begin{eqnarray}\label{eq-asym-Bi}
\begin{cases}
\Bi(Z)\sim
\frac{i}{2\sqrt{\pi}}Z^{-\frac{1}{4}}e^{-\frac{2}{3}Z^{\frac{3}{2}}}
+\frac{1}{\sqrt{\pi}}Z^{-\frac{1}{4}}e^{\frac{2}{3}Z^{\frac{3}{2}}},
\quad&\arg{Z}\in\left(-\frac{\pi}{3},\pi\right),\\
\Bi(Z)\sim
\frac{i}{2\sqrt{\pi}}Z^{-\frac{1}{4}}e^{-\frac{2}{3}Z^{\frac{3}{2}}}
+\frac{1}{2\sqrt{\pi}}Z^{-\frac{1}{4}}e^{\frac{2}{3}Z^{\frac{3}{2}}},
\quad & \arg{Z}\in\left(\frac{\pi}{3},\frac{5\pi}{3}\right),\\
\Bi(Z)\sim
-\frac{i}{2\sqrt{\pi}}Z^{-\frac{1}{4}}e^{-\frac{2}{3}Z^{\frac{3}{2}}}
+\frac{1}{\sqrt{\pi}}Z^{-\frac{1}{4}}e^{\frac{2}{3}Z^{\frac{3}{2}}},
\quad & \arg{Z}\in\left(-\pi,\frac{\pi}{3}\right).
\end{cases}
\end{eqnarray}

When $\lambda\to\infty$ with $\arg\lambda\sim\frac{\pi}{5}$, it follows from $\lambda=\xi^{\frac{2}{5}}z$  and \eqref{eq-asymp-relation-zeta-z} that $\arg\zeta\sim\frac{\pi}{3}$. Substituting \eqref{eq-asymp-relation-zeta-z} into \eqref{eq-asym-Ai} and \eqref{eq-asym-Bi},
and using $\lambda=\xi^{\frac{2}{5}}z$ again, we get
\begin{eqnarray}\label{eq-Ai-Bi-pi/5}
\begin{cases}
\left(\frac{\zeta}{f(z)}\right)^{\frac{1}{4}}\Ai_{n}(\xi^{\frac{2}{3}}\zeta)=\left(\frac{\zeta}{f(z)}\right)^{\frac{1}{4}}\left(\frac{d\hat{\zeta}}{d\zeta}\right)^{-\frac{1}{2}}\Ai\left(\xi^{\frac{2}{3}}\hat{\zeta}\right)\sim c_{1}\frac{-1}{\sqrt{2}}\lambda^{-\frac{3}{4}}e^{-\frac{4}{5}\lambda^{\frac{5}{2}}},\\
\left(\frac{\zeta}{f(z)}\right)^{\frac{1}{4}}\Bi_{n}(\xi^{\frac{2}{3}}\zeta)=\left(\frac{\zeta}{f(z)}\right)^{\frac{1}{4}}\left(\frac{d\hat{\zeta}}{d\zeta}\right)^{-\frac{1}{2}}\Bi\left(\xi^{\frac{2}{3}}\hat{\zeta}\right)\sim ic_{1}\frac{-1}{\sqrt{2}}\lambda^{-\frac{3}{4}}e^{-\frac{4}{5}\lambda^{\frac{5}{2}}}+2c_{2}\frac{1}{\sqrt{2}}\lambda^{-\frac{3}{4}}e^{\frac{4}{5}\lambda^{\frac{5}{2}}}
\end{cases}
\end{eqnarray}
as $\lambda\to\infty$ (and $\zeta\to\infty$ accordingly), where
\begin{equation}
\begin{split}
c_{1}&=\frac{-1}{2\sqrt{\pi}}\xi^{-\frac{2}{15}}\exp\left\{-\sum\limits_{s=0}^{n}\alpha_{s}(A(\xi),B(\xi))\xi^{-s+1}\right\},\\
c_{2}&=\frac{1}{2\sqrt{\pi}}\xi^{-\frac{2}{15}}\exp\left\{\sum\limits_{s=0}^{n}\alpha_{s}(A(\xi),B(\xi))\xi^{-s+1}\right\}.
\end{split}
\end{equation}
From \eqref{eq-canonical-solutions}, a straightforward calculation yields
\begin{equation}\label{eq-asym-Phi21-Phi22}
\left((\hat{\Phi}_{k})_{21},(\hat{\Phi}_{k})_{22}\right)\sim \left(\frac{1}{\sqrt{2}}\lambda^{-\frac{3}{4}}e^{\frac{4}{5}\lambda^{\frac{5}{2}}},\,
\frac{-1}{\sqrt{2}}\lambda^{-\frac{3}{4}}e^{-\frac{4}{5}\lambda^{\frac{5}{2}}}\right),
\qquad k\in\mathbb{Z}
\end{equation}
as $\lambda\to\infty$.
Substituting \eqref{eq-Ai-Bi-pi/5} into \eqref{eq-higher-approximation-ODE-1}, noting that $Y=\left(\frac{\zeta}{f(z)}\right)^{\frac{1}{4}}$ and $Y$ is the first component of $\hat{\Phi}$,
and then comparing the resulting equation with \eqref{eq-asym-Phi21-Phi22},
we get
\begin{equation}\label{eq-Y-pi/5}
\begin{split}
W=&[C_{1}+r_{1,1}(\xi)]c_{1}(\hat{\Phi}_{1})_{22}+[C_{2}+r_{2,1}(\xi)][ic_{1}(\hat{\Phi}_{1})_{22}+2c_{2}(\hat{\Phi}_{1})_{21}]\\
=&[C_{2}+r_{2,1}(\xi)]2c_{2}(\hat{\Phi}_{1})_{21}+\{[C_{1}+r_{1,1}(\xi)]c_{1}+[C_{2}+r_{2,1}(\xi)]ic_{1}\}(\hat{\Phi}_{1})_{22}
\end{split}
\end{equation}
with $\lambda\in\Omega_{1}$.

When $\lambda\to\infty$ with $\arg\lambda\sim\frac{3\pi}{5}$,
one has $\arg{\zeta}\sim\pi$. Using the asymptotics of the Airy functions with $\arg{Z}\sim\pi$ in \eqref{eq-asym-Ai} and \eqref{eq-asym-Bi}, we have
\begin{eqnarray}\label{eq-Ai-Bi-3pi/5}
\begin{cases}
\left(\frac{\zeta}{f(z)}\right)^{\frac{1}{4}}\Ai_{n}(\xi^{\frac{2}{3}}\zeta)=\left(\frac{\zeta}{f(z)}\right)^{\frac{1}{4}}\left(\frac{d\hat{\zeta}}{d\zeta}\right)^{-\frac{1}{2}}\Ai\left(\xi^{\frac{2}{3}}\hat{\zeta}\right)\sim c_{1}\frac{-1}{\sqrt{2}}\lambda^{-\frac{1}{4}}e^{-\frac{4}{5}\lambda^{\frac{5}{2}}}+ic_{2}\frac{1}{\sqrt{2}}\lambda^{-\frac{1}{4}}e^{\frac{4}{5}\lambda^{\frac{5}{2}}},\\
\left(\frac{\zeta}{f(z)}\right)^{\frac{1}{4}}\Bi_{n}(\xi^{\frac{2}{3}}\zeta)=\left(\frac{\zeta}{f(z)}\right)^{\frac{1}{4}}\left(\frac{d\hat{\zeta}}{d\zeta}\right)^{-\frac{1}{2}}\Bi\left(\xi^{\frac{2}{3}}\hat{\zeta}\right)\sim ic_{1}\frac{-1}{\sqrt{2}}\lambda^{-\frac{1}{4}}e^{-\frac{4}{5}\lambda^{\frac{5}{2}}}+c_{2}\frac{1}{\sqrt{2}}\lambda^{-\frac{1}{4}}e^{\frac{4}{5}\lambda^{\frac{5}{2}}}
\end{cases}
\end{eqnarray}
as $\lambda\to\infty$ (and $\zeta\to\infty$ accordingly).
In a similar manner for deriving \eqref{eq-Y-pi/5}, we get
\begin{equation}\label{eq-Y-3pi/5-C}
\begin{split}
W=&\{[C_{1}+r_{1,2}(\xi)]ic_{2}+[C_{2}+r_{2,2}(\xi)]c_{2}\}(\hat{\Phi}_{2})_{21}\\
&+\{[C_{1}+r_{1,2}(\xi)]c_{1}+[C_{2}+r_{2,2}(\xi)]ic_{1}\}(\hat{\Phi}_{2})_{22}
\end{split}
\end{equation}
$\lambda\in\Omega_{2}$.
Combining \eqref{eq-Y-pi/5} with \eqref{eq-Y-3pi/5-C}, and noting that
$$\left((\hat{\Phi}_{2})_{21},(\hat{\Phi}_{2})_{22}\right)=\left((\hat{\Phi}_{1})_{21},(\hat{\Phi}_{1})_{22}\right)\left(\begin{matrix}1&s_{1}\\0&1\end{matrix}\right),$$
we further obtain
\begin{equation}
r_{1,1}(\xi)-r_{1,2}(\xi)=i(r_{2,2}(\xi)-r_{2,1}(\xi))
\end{equation}
and
\begin{equation}
\begin{split}
s_{1}&=\frac{(C_{2}-iC_{1})c_{2}+(2r_{2,1}(\xi)-ir_{1,2}(\xi)-r_{2,2}(\xi))c_{2}}{(C_{1}+iC_{2})c_{1}+r_{1,2}(\xi)c_{1}+r_{2,2}(\xi)ic_{1}}\\
&=-\frac{ic_{2}}{c_{1}}+\frac{2(r_{2,1}(\xi)-r_{2,2}(\xi))c_{2}}{(C_{1}+iC_{2}+r_{1,2}(\xi)+ir_{2,2}(\xi))c_{1}}.
\end{split}
\end{equation}
By Lemma~\ref{lem-higher-order-approximation-ODE-1}, $r_{j,k}(\xi)=\mathcal{O}\left(\frac{|C_{1}|+|C_{2}|}{\xi^{n}}\right)$ as $\xi\to+\infty$ for $j=1,2$ and $k=1,2$. Hence,
\begin{equation}
s_{1}=i \exp\left\{\sum\limits_{s=0}^{2n}2\alpha_{s}(A(\xi),B(\xi))\xi^{-s+1}\right\}\left(1+\mathcal{O}(\xi^{-2n})\right)
\end{equation}
as $\xi\to+\infty$.

The asymptotic expansion of $s_{-1}$ in \eqref{eq-stokes-full-asymptotic-thm1}
follows from the conjugate relation $s_{-1}=\overline{s_{1}}$, and the asymptotic expansions for other stokes multipliers follow from the constraint $s_{k}=i(1+s_{k+2}s_{k+3})$.

\

\textbf{Proof of Theorem \ref{Thm-s-large-pole}}

The full asymptotic expansion of $s_{1}$ in Theorem \ref{Thm-s-large-pole} can be derived similarly to those in Theorem~\ref{Thm-s-large-value}. However, one shoul note that the asymptotic behavior of $\hat{\Phi}(\lambda,t)$ as $\lambda\to\infty$ has been varied when $t\to p$ (compare \eqref{eq-canonical-solutions-hat-Phi} and \eqref{eq-canonical-solutions-near-pole}). Hence, when as $\lambda\to\infty$ and $t=p$, we should replace \eqref{eq-Ai-Bi-pi/5} and \eqref{eq-asym-Phi21-Phi22} by
\begin{eqnarray}\label{eq-Ai-Bi-pi/5-2}
\begin{cases}
\left(\frac{\eta}{\tilde{f}(z)}\right)^{\frac{1}{4}}\Ai_{n}(\xi^{\frac{2}{3}}\eta)=\left(\frac{\eta}{\tilde{f}(z)}\right)^{\frac{1}{4}}\left(\frac{d\hat{\eta}}{d\eta}\right)^{-\frac{1}{2}}\Ai\left(\xi^{\frac{2}{3}}\hat{\eta}\right)\sim \tilde{c}_{1}\frac{-i}{\sqrt{2}}\lambda^{-\frac{3}{4}}e^{-\frac{4}{5}\lambda^{\frac{5}{2}}},\\
\left(\frac{\eta}{\tilde{f}(z)}\right)^{\frac{1}{4}}\Bi_{n}(\xi^{\frac{2}{3}}\eta)=\left(\frac{\eta}{\tilde{f}(z)}\right)^{\frac{1}{4}}\left(\frac{d\hat{\eta}}{d\eta}\right)^{-\frac{1}{2}}\Bi\left(\xi^{\frac{2}{3}}\hat{\eta}\right)\sim i\tilde{c}_{1}\frac{-i}{\sqrt{2}}\lambda^{-\frac{3}{4}}e^{-\frac{4}{5}\lambda^{\frac{5}{2}}}+2\tilde{c}_{2}\frac{-i}{\sqrt{2}}\lambda^{-\frac{3}{4}}e^{\frac{4}{5}\lambda^{\frac{5}{2}}}
\end{cases}
\end{eqnarray}
and
\begin{equation}\label{eq-asym-Phi21-Phi22-near-pole}
\left((\hat{\Phi}_{k})_{21},(\hat{\Phi}_{k})_{22}\right)\sim \left(\frac{-i}{\sqrt{2}}\lambda^{-\frac{3}{4}}e^{\frac{4}{5}\lambda^{\frac{5}{2}}},\,
\frac{-i}{\sqrt{2}}\lambda^{-\frac{3}{4}}e^{-\frac{4}{5}\lambda^{\frac{5}{2}}}\right),
\qquad k\in\mathbb{Z}
\end{equation}
respectively,
where
\begin{equation}
\begin{split}
\tilde{c}_{1}&=\frac{i}{2\sqrt{\pi}}\xi^{-\frac{2}{15}}\exp\left\{-\sum\limits_{s=0}^{n}\beta_{s}\xi^{-2s+1}\right\},\\
\tilde{c}_{2}&=\frac{i}{2\sqrt{\pi}}\xi^{-\frac{2}{15}}\exp\left\{\sum\limits_{s=0}^{n}\beta_{s}\xi^{-2s+1}\right\}.
\end{split}
\end{equation}
Hence, by a similar way to derive $s_{1}$ in the proof of Theorem \ref{Thm-s-large-value}, we obtain, for any $n\in\mathbb{N}$, that
\begin{equation}
s_{1}=-i \exp\left\{\sum\limits_{s=0}^{2n}2\beta_{s}\xi^{-2s+1}\right\}\left(1+\mathcal{O}(\xi^{-2n})\right)
\end{equation}
as $\xi\to+\infty$.

Note that the Stokes multipliers corresponding to the real tritronqu\'{e}e solution of PI equation are $s_{1}=s_{-1}=i, s_{2}=s_{-2}=0, s_{0}=i$. This implies that $\beta_{s}$ are purely imaginary for all $s \in \mathbb{N}$. Making use of this fact and the constrain $s_{k} = i(1 + s_{k+2}s_{k+3})$, we obtain the full asymptotic expansion of the other Stokes multipliers.

\appendix

\section{Monodromy theory of PI}
\label{sec:AppA}

Painlev\'e equations can be expressed as compatibility conditions of
Lax pairs. In particular, consider the Lax pair of PI (see \cite{Kapaev-Kitaev-1993})
\begin{equation}\label{lax pair-I}
\left\{\begin{aligned}
\frac{\partial\Psi}{\partial\lambda}
&=\left\{(4\lambda^4+t+2y^2)\sigma_{3}-i(4y\lambda^2+t+2y^2)\sigma_{2}
-\left(2y_{t}\lambda+\frac{1}{2\lambda}\right)\sigma_{1}\right\}\Psi
:=\mathcal{A}(\lambda)\Psi,
\\
\frac{\partial\Psi}{\partial t}
&=\left\{\left(\lambda+\frac{y}{\lambda}\right)\sigma_{3}-\frac{iy}{\lambda}\sigma_{2}\right\}\Psi
:=\mathcal{B}(\lambda)\Psi,
\end{aligned}\right.
\end{equation}
where
\[
\sigma_{1}=\left(\begin{matrix}0&1\\1&0\end{matrix}\right),\quad \sigma_{2}=\left(\begin{matrix}0&-i\\i&0\end{matrix}\right),\quad \sigma_{3}=\left(\begin{matrix}1&0\\0&-1\end{matrix}\right)
\]
are the Pauli matrices and $y_{t}=\frac{dy}{dt}$.
If the Lax pair satisfies the compatibility condition $\frac{\partial^2\Psi}{\partial t\partial\lambda}=\frac{\partial^2\Psi}{\partial\lambda\partial t}$,
then $y=y(t)$ is a solution of the PI equation~\eqref{PI equation}.

Under the transformation
\begin{equation}\label{eq-transform-canonical solution}
\Phi(\lambda)=\lambda^{\frac{1}{4}\sigma_{3}}\frac{\sigma_{3}+\sigma_{1}}{\sqrt{2}}\Psi(\sqrt{\lambda}),
\end{equation}
the first equation of (\ref{lax pair-I}) becomes
\begin{equation}\label{eq-fold-Lax-pair}
\frac{\partial\Phi}{\partial\lambda}=\left(\begin{matrix}y_{t}&2\lambda^{2}+2y\lambda-t+2y^2\\2(\lambda-y)&-y_{t}\end{matrix}\right)\Phi.
\end{equation}
Following \cite{Kapaev-Kitaev-1993} (see also \cite{AAKapaev-2004}), the only singularity of equation (\ref{eq-fold-Lax-pair}) is the  irregular singular point at $\lambda=\infty$, and there exist canonical solutions $\Phi_{k}(\lambda)$, $k\in\mathbb{Z}$, of (\ref{eq-fold-Lax-pair}) with the following asymptotic expansion
\begin{equation}\label{eq-canonical-solutions}
\Phi_{k}(\lambda,t)
=\lambda^{\frac{1}{4}\sigma_{3}}\frac{\sigma_{3}+\sigma_{1}}{\sqrt{2}}
\left(I+\frac{\mathcal{H}}{\lambda}+\mathcal{O}\left(\frac{1}{\lambda^2}\right)\right)e^{(\frac{4}{5}\lambda^{\frac{5}{2}}+t\lambda^{\frac{1}{2}})\sigma_{3}},
\qquad\lambda\rightarrow\infty,\quad \lambda\in\Omega_{k},
\end{equation}
uniformly for all $t$ bounded away from poles, where $\mathcal{H}=-(\frac{1}{2}y_{t}^2-2y^3-ty)\sigma_{3}$, and the canonical sectors are
$$\Omega_{k}=\left\{\lambda\in\mathbb{C}:~\arg \lambda\in \left(-\frac{3\pi}{5}+\frac{2k\pi}{5},\frac{\pi}{5}+\frac{2k\pi}{5}\right)\right\}, \qquad k\in\mathbb{Z}.$$
These canonical solutions are related by
\begin{equation}\label{eq-Stokes-matrices}
\Phi_{k+1}=\Phi_{k}S_{k},\quad S_{2k-1}=\left(\begin{matrix}1&s_{2k-1}\\0&1\end{matrix}\right),\quad S_{2k}=\left(\begin{matrix}1&0\\s_{2k}&1\end{matrix}\right),
\end{equation}
where $s_{k}$'s are the \emph{Stokes multipliers}, and independent of $\lambda$ and $t$ according to the isomonodromy condition. The Stokes multipliers are subject to the constraints
\begin{equation}\label{eq-constraints-stokes-multipliers}
s_{k+5}=s_{k}\quad \text{and}\quad s_{k}=i(1+s_{k+2}s_{k+3}),
\qquad k\in\mathbb{Z}.
\end{equation}
Moreover, regarding $s_{k}$'s as functions of $(t,y(t),y'(t))$,   they also satisfy (see~\cite[eq.~(13)]{AAKapaev-1988})
\begin{equation}\label{eq-sk-s-k-relation}
s_{k}\left (t,y(t),y'(t)\right)=-\overline{s_{-k}\left (\bar{t},\overline{y(t)},\overline{y'(t)}\right )}, \quad k\in\mathbb{Z},
\end{equation}
where $\bar{z}$ stands for the complex conjugate of $z$.
From~\eqref{eq-constraints-stokes-multipliers}, it is readily seen that, generically, two of the Stokes multipliers determine all others. The derivation of (\ref{eq-canonical-solutions}), (\ref{eq-Stokes-matrices}), and (\ref{eq-constraints-stokes-multipliers}), along with more details about the Lax pairs, can be found in~\cite{FAS-2006}.

For simplicity, we make a further transformation:
\begin{equation}\label{eq-def-hat{Phi}}
\hat{\Phi}(\lambda,t)=G(\lambda,t)\Phi(\lambda,t)
\end{equation}
with
\begin{equation}\label{def-G(lambda,t)}
G(\lambda,t)
=\begin{bmatrix}0&1\\ 1&-\frac{1}{2}\left(-y_{t}+\frac{1}{2(\lambda-y)}\right)\end{bmatrix}
(\lambda-y)^{\frac{\sigma_{3}}{2}}.
\end{equation}
Then, $\hat{\Phi}(\lambda,t)$ satisfies
\begin{equation}\label{eq-system-hat-Phi}
\frac{d}{d\lambda}\hat{\Phi}(\lambda,t)
=\begin{bmatrix}0&2\\V(\lambda,t)&0\end{bmatrix}\hat{\Phi}(\lambda,t),
\end{equation}
where
\[
2V(\lambda,t)=\left\{\begin{aligned}
&y_{t}^{2}+4\lambda^{3}+2\lambda t-2y t-4y^{3}-\frac{y_{t}}{\lambda-y}+\frac{3}{4}\frac{1}{(\lambda-y)^2}, \qquad &t\neq p,\\
&4\lambda^3+2p\lambda-28H,\qquad &t=p,
\end{aligned}\right.
\]
$p$ is a pole of the PI solution and $H$ is the parameter in the Laurent series~\eqref{eq-Laurent-series}.

Now, the approximation of $\Phi(\lambda,t)$ in \eqref{eq-canonical-solutions} is transformed to
\begin{equation}\label{eq-canonical-solutions-hat-Phi}
\hat{\Phi}_{k}(\lambda,t)
=\frac{\lambda^{-\frac{3}{4}\sigma_{3}}}{\sqrt{2}}
\begin{bmatrix}1&-1\\1&1\end{bmatrix}
\left(I+\mathcal{O}(\lambda^{-\frac{1}{2}})\right)
e^{(\frac{4}{5}\lambda^{\frac{5}{2}}+t\lambda^{\frac{1}{2}})\sigma_{3}}
\end{equation}
as $\lambda\rightarrow\infty$ with $\lambda\in\Omega_{k}$ when $t$ is bounded away from $p$.
When $t\to p$, according to \cite{Bertola-Tovbis-2013}, the above approximation \eqref{eq-canonical-solutions-hat-Phi} should be replaced by
\begin{equation}\label{eq-canonical-solutions-near-pole}
\hat{\Phi}_{k}(\lambda,p)
=\frac{\lambda^{-\frac{3}{4}\sigma_{3}}}{\sqrt{2}}
\begin{bmatrix}-i&-i\\-i&i\end{bmatrix}
\left(I+\mathcal{O}(\lambda^{-\frac{1}{2}})\right)
e^{(\frac{4}{5}\lambda^{\frac{5}{2}}+p\lambda^{\frac{1}{2}})\sigma_{3}}
\end{equation}
as $\lambda\rightarrow\infty$ with $\lambda\in\Omega_{k}$; see \cite[Corollary A.8]{Bertola-Tovbis-2013}.

\section{Asymptotics of PI solutions as $t\to-\infty$}
\label{sec:AppB}

The asymptotic behavior of the PI solution as $t\to-\infty$ has been extensively studied by
Kapaev~\cite{AAKapaev-1988} using the monodromy theory.
There are three types of real solutions for the PI equation classified by \eqref{eq-classifed-by-stokes}. They satisfy
the following asymptotic behavior:
\begin{enumerate}[(A)]
\item a two-parameter family of solutions, oscillating about the parabola $y=-\sqrt{-t/6}$ and satisfying
\begin{equation}\label{eq-behavior-type-A}
y=-\sqrt{-\frac{t}{6}}+d\,(-t)^{-\frac{1}{8}}\cos{\left[24^{\frac{1}{4}}\left(\frac{4}{5}(-t)^{\frac{5}{4}}-\frac{5}{8}d^2 \log(-t)+\varphi\right)\right]}+\mathcal{O}\left (t^{-\frac{5}{8}}\right )
\end{equation}
as $t\rightarrow-\infty$, where
    \begin{equation}\label{eq-parameter-d-theta}
  \left\{\begin{aligned}
  &24^{\frac{1}{4}}d^{2}=-\frac{1}{\pi}\log{|s_{0}|},\\
  &24^{\frac{1}{4}}\theta=-\arg{s_{3}}-24^{\frac{1}{4}}d^{2}\left(\frac{19}{8}\log{2}+\frac{5}{8}\log{3}\right)-\frac{\pi}{4}-\arg\Gamma\left(-i\frac{24^{\frac{1}{4}}}{2}d^{2}\right);
\end{aligned}\right.
\end{equation}
\item a one-parameter family of solutions (termed {\it separatrix solutions}), satisfying
    \begin{equation}\label{eq-behavior-type-B}
      y(t)=y_{0}(t)-\frac{h}{4\sqrt{\pi}}24^{-\frac{1}{8}}(-t)^{-\frac{1}{8}}\exp\left\{-\frac{4}{5}24^{\frac{1}{4}}(-t)^{\frac{5}{4}}\right\}\left(1+\mathcal{O}\left(|t|^{-\frac{5}{4}}\right)\right)
    \end{equation}
    as $t\rightarrow-\infty$, where $y_{0}(t)=\sqrt{\frac{-t}{6}}\left[1+\mathcal{O}\left((-t)^{-\frac{5}{2}}\right)\right]$ and
    \begin{equation}\label{eq-parameter-h}
     h=s_{1}-s_{4};
    \end{equation}
\item a two-parameter family of solutions, having infinitely many double poles on the negative real axis and satisfying
    \begin{equation}\label{eq-behavior-type-C}
      \frac{1}{y(t)+\sqrt{{-t}/{6}}}\sim \frac{\sqrt{6}}{2}\sin^{2}\left\{\frac{2}{5}24^{1/4}(-t)^{\frac{5}{4}}+\frac{5}{8}\rho\log(-t)+\sigma\right\}\quad \text{as}\quad t\rightarrow-\infty,
    \end{equation}
    where
    \begin{equation}\label{eq-parameter-rho-sigma}
    \left\{\begin{aligned}
    \rho&=\frac{1}{2\pi}\log(|s_{2}|^{2}-1)=\frac{1}{2\pi}\log(|1+s_{2}s_{3}|)=\frac{1}{2\pi}\log|s_{0}|,\\
    \sigma&=\frac{19}{8}\rho\log{2}+\frac{5}{8}\rho\log{3}+\frac{1}{2}\arg\Gamma\left(\frac{1}{2}-i\rho\right)-\frac{\pi}{4}+\frac{1}{2}\arg{s_{2}}.
    \end{aligned}\right.
    \end{equation}
\end{enumerate}

\section{Proof of Remark \ref{remark-tilde-h}}
\label{tilde-h-odd-term-vanish}

If, in eq.~\eqref{eq-tilde-h-behave-like}, we assume that $\tilde{h}(\xi)$ behaves as
\begin{equation}\label{eq-tilde-h-behave-like-appendix}
\tilde{h}(\xi)=\sum\limits_{s=0}^{n-1}\frac{\tilde{h}_{s}}{\xi^{2s}}+\sum\limits_{s=1}^{n-1}\frac{\bar{h}_{s}}{\xi^{2s-1}}+\mathcal{O}(\xi^{-2n})
\end{equation}
and at least one of $\bar{h}_{s}$ is nonzero. Without loss of generality, we may assume $\bar{h}_{1}\neq 0$. Under this assumption, \eqref{eq-tilde-An-def} should be replaced by
\begin{equation}\label{eq-tilde-An-def-2}
\mathcal{\tilde{A}}_{n}(\eta,\xi)=\sum\limits_{s=1}^{n}\frac{\tilde{a}_{s}(\eta)}{\xi^{2s}}+\sum\limits_{s=2}^{n}\frac{\bar{a}_{s}(\eta)}{\xi^{2s-1}}.
\end{equation}
By comparing the coefficients of $\xi^{-s}, s=0,1,\cdots,2n$ on both sides of \eqref{eq-tilde-An-coeff-compare}, and using the analyticity property, we can recursively determine all of $\tilde{a}_{s}(\eta)$ for $s=1,\dots,n$ and $\bar{a}_{s}(\eta)$ for $s=2,\dots,n$. In particular, we have
\begin{equation}\label{eq-tilde-bar-a-integral}
\begin{split}
\tilde{a}_{1}(\eta)&=\frac{1}{2\eta^{\frac{1}{2}}}\int_{0}^{\eta}\frac{\tilde{\psi}(t)}{t^{\frac{1}{2}}}dt+\frac{1}{2\eta^{\frac{1}{2}}}\int_{z_{1}}^{z}\frac{4\tilde{h}_{0}}{\tilde{f}(z)^{\frac{1}{2}}}dz,\\
\bar{a}_{2}(\eta)&=\frac{1}{2\eta^{\frac{1}{2}}}\int_{z_{1}}^{z}\frac{4\bar{h}_{1}}{\tilde{f}(z)^{\frac{1}{2}}}dz,
\end{split}
\end{equation}
where $\tilde{f}(z)$ and $\tilde{\psi}(\eta)$ are defined in \eqref{eq-def-tilde-f} and \eqref{eq-def-tilde-psi}, respectively.
Following the analysis in Section~\ref{sec:proof-theorem}, we derive the asymptotic behavior of the Stokes multiplier $s_{1}$ in the form
\begin{equation}\label{eq-s1-asym-appendix}
s_{1}\sim -i\exp\left\{\sum\limits_{s=0}^{\infty}\frac{2\beta_{s}}{\xi^{2s-1}}+\sum\limits_{s=2}^{\infty}\frac{2\bar{\beta}_{s}}{\xi^{2s-2}}\right\}
\end{equation}
as $\xi\to\infty$, where $\beta_{s}=\lim\limits_{\eta\to\infty}\eta^{\frac{1}{2}}\tilde{a}_{s}(\eta)$ and $\bar{\beta}_{s}=\lim\limits_{\eta\to\infty}\eta^{\frac{1}{2}}\bar{a}_{s}(\eta)$. Observing that $\eta^{\frac{1}{2}}d\eta=\tilde{f}(z)^{\frac{1}{2}}dz$, we have
\begin{equation}\label{eq-beta-limit-expression}
\beta_{1}=\frac{1}{2}\int_{z_{1}}^{\infty e^{i\theta}}\frac{\tilde{\psi}(\eta(z))}{\eta(z)}\tilde{f}(z)^{\frac{1}{2}}dz+\frac{1}{2\eta^{\frac{1}{2}}}\int_{z_{1}}^{\infty e^{i\theta}}\frac{4\tilde{h}_{0}}{\tilde{f}(z)^{\frac{1}{2}}}dz
\end{equation}
and
\begin{equation}\label{eq-bar-beta-limit-expression}
\bar{\beta}_{2}=\frac{1}{2\eta^{\frac{1}{2}}}\int_{z_{1}}^{\infty e^{i\theta}}\frac{4\tilde{h}_{1}}{\tilde{f}(z)^{\frac{1}{2}}}dz,
\end{equation}
where $\theta\in[0,\frac{4\pi}{5}]$.
From \eqref{eq-beta-limit-expression}, we see that, by appropriately choosing $\tilde{h}_{0}$, it is possible to ensure that $\Re{\beta}_{1}=0$. However, since $\tilde{h}_{1}\neq 0$, it follows that $\Re{\bar{\beta}_{2}}\neq 0$. Therefore, the asymptotic approximation of $s_{1}$ in \eqref{eq-s1-asym-appendix} contradicts the fact that $s_{1}=i$ for the real tritronqu\'{e}e solution of the PI equation. Therefore, all the coefficients $\bar{h}_{s}$ of the odd terms in \eqref{eq-tilde-h-behave-like-appendix} must vanish.

\section*{Acknowledgements}
All authors are grateful to the two anonymous reviewers for their valuable comments, which significantly improved the paper.

The work of Wen-Gao Long was partially supported by the National Natural Science Foundation of China [Grant No. 12401094], the Natural Science Foundation of Hunan Province [Grant No. 2024JJ5131] and the Outstanding Youth Fund of Hunan Provincial Department of Education [Grant No. 23B0454]. The work of Yu-Tian Li was supported in part by the National Natural Science Foundation of China
[Grant no. 11801480].

\end{document}